\documentclass[twocoloumn]{IEEEtran}
\usepackage[utf8]{inputenc}
\usepackage{nomencl}
\usepackage{cite}
\usepackage[dvips]{graphicx,psfrag}
\usepackage{amsthm}
\usepackage{amsmath}
\usepackage{amsmath}
\usepackage{amssymb}
\usepackage{rotating}
\usepackage{subfigure}
\usepackage{cite}
\usepackage{color}
\usepackage{stfloats}
\newtheorem{theorem}{Theorem}
\newtheorem{lemma}{Lemma}
\newtheorem{corollary}{Corollary}

\usepackage{hyperref}
\usepackage{algorithmic}
\usepackage[linesnumbered,ruled,vlined]{algorithm2e}
\usepackage{mathtools}
\usepackage {tikz}
\usetikzlibrary {positioning}
\definecolor {processblue}{cmyk}{0.96,0,0,0}
\allowdisplaybreaks

\newcommand\eqa{\mathrel{\overset{\makebox[0pt]{\mbox{\normalfont\tiny\sffamily (a)}}}{=}}}
\newcommand\eqb{\mathrel{\overset{\makebox[0pt]{\mbox{\normalfont\tiny\sffamily (b)}}}{=}}}
\newcommand\eqc{\mathrel{\overset{\makebox[0pt]{\mbox{\normalfont\tiny\sffamily (c)}}}{=}}}
\newcommand\eqd{\mathrel{\overset{\makebox[0pt]{\mbox{\normalfont\tiny\sffamily (d)}}}{=}}}
\newcommand\eqe{\mathrel{\overset{\makebox[0pt]{\mbox{\normalfont\tiny\sffamily (e)}}}{=}}}

\begin{document}
\title{Load-aware Performance Analysis of Cell Center/Edge Users in Random HetNets}
\author{Praful D. Mankar, Goutam Das,~\IEEEmembership{Member,~IEEE} and S. S. Pathak,~\IEEEmembership{Senior Member,~IEEE}% <-this % stops a space
\thanks{Praful D. Mankar and S. S. Pathak are with the Dept. of E \& ECE; and Goutam Das is with the G. S. Sanyal School of Telecommunications, Indian Institute of Technology, Kharagpur, WB, India. E-mail: \{praful,ssp\}@ece.iitkgp.ernet.in and gdas@gssst.iitkgp.ernet.in}% <-this % stops a space
% \thanks{Goutam Das is with the G.S.S.S.T., Indian Institute of Technology, Kharagpur,
% WB, India. Email: gdas@gssst.iitkgp.ernet.in}
}

\maketitle
\begin{abstract}
% In recent years, a practice of modeling tiers of heterogeneous networks (macro, femto etc.) using independent Poisson point processes is being followed to include the spatial randomness into the coverage analysis. 
For real-time traffic, the link quality and call blocking probability (both derived from coverage probability) are realized to be poor for cell edge users (CEUs) compared to cell center users (CCUs) as the signal reception in the cell center region is better compared to the cell edge region. In heterogeneous networks (HetNets), the uncoordinated channel access by different types of base stations determine the interference statistics that further arbitrates the coverage probability. Thus, the spectrum allocation techniques have major impact on the performance of CCU and CEU. In this paper, the performance of CCUs and CEUs in a random two-tier network is studied for two spectrum allocation techniques namely: 1) co-channel (CSA), and 2) shared (SSA).  For performance analysis, the widely accepted conception of modeling the tiers of HetNet using independent homogeneous Poisson point process (PPP) is considered to accommodate the spatial randomness in location of BSs.
To incorporate the spatial randomness in the arrival of service and to aid the load-aware analysis, the cellular traffic is modeled using spatio-temporal PPP. Under this scenario, we have developed an analytical framework to evaluate the load-aware performance, including coverage and blocking probabilities, of CCUs and CEUs under both spectrum allocation techniques. 
Further, we provide insight into achievable area energy efficiency for SSA and CSA. The developed analytical framework is validated through extensive simulations. Next, we demonstrate the impact of traffic load and femto access points density on the performance of CCUs/CEUs under CSA and SSA. %We observe that a properly configured SSA can provide better coverage probability along with improved blocking probability as compared to CSA.
% This letter presents the coverage analysis of a macro user in two-tier heterogeneous network under shared co-channel spectrum allocation. Considering interference limited cell edge macro users link reliability, distinct set of channels are reserved for them and femto access points (FAPs) are confined to access the orthogonal set of channels so that the inter-tier interference for cell edge macro users can be avoided. In addition, our analysis also unveils the impact of the real-time traffic intensity on the coverage probability. Further, we provide the numerical results demonstrating the effectiveness of proposed spectrum allocation over the co-channel spectrum allocation in overall coverage probability and mean rate.
\end{abstract}
\IEEEpeerreviewmaketitle
\section{Introduction}
\IEEEPARstart{C}{ellular} network has been evolving since last few decades to meet the requirement of broadband connectivity, reliable communication, mobility, etc. The process resulted in inclusion of various types of base stations (like femto, pico, micro, etc.) differing in transmission power along with existing cellular/macro base stations (BSs). Because of heterogeneity in deployment, transmission power and functionality such networks are generally referred as heterogeneous networks (HetNets) or multi-tier networks. 
For performance evaluation of traditional cellular networks, the researchers mostly relied on highly approximated models like Wyner model \cite{Wyner}, grid-based hexagonal model \cite{Rappaport}, etc. However, the modeling of location of base stations (BSs) using homogeneous Poisson point process (PPP) is considered in \cite{Baccelli1997,Brown} for performance evaluation of a realistic scenario. Under this type of modeling, the seminal result on coverage probability in cellular networks produced in \cite{Andrews_2011}. Furthermore, on demand basis deployment of small BSs inherently introduces spatial randomness in a multi-tier network. Therefore, using the analysis of \cite{Andrews_2011} and applications of stochastic geometry, the practice of modeling tiers of a cellular network using independent homogeneous PPPs is being followed \cite{Dhillon_2011,Dhillon2012,HSJo_2012,Mukherjee_2011,Heath_2013}. This brings analytical tractability along with practical relevance into the analysis.  The interested reader may refer to \cite{Haenggi_Book,haenggi2009interference,Baccelli,Jeffrey_2016}.

Usually in HetNets, deployment of macro tier stands out for low bandwidth services (like VoIP) and ubiquitous link connectivity. However,  the link quality of a mobile user is subjected to its location. The users in the cell edge region usually receive weaker signal strength from the serving macro BS (MBS) and stronger interference from co-channel MBSs. Besides, massive inter-tier interference is imposed by subscriber-owned access points (such as fetmo cells). This leads the cell edge users (CEUs) to experience reduced coverage which further degrades their achievable transmission rate. %the key performance indicators such as average rate, delay, call admission probability etc. 
However, interference from low powered femto access points (FAPs) and co-channel MBSs to cell center users (CCUs) is significantly below the signal strength received from its associated MBS. Therefore, better coverage can be experienced in the cell center region compared to the cell edge region. 
Further, the lower transmission rate in cell edge region makes CEUs to be bandwidth hungry. Thus, the blocking probability experienced by a CEU is significantly higher than that of a CCU. 
% In summery, in cellular network, the CCUs experience lower blocking  along with better coverage; on the contrary CEUs are subjected to higher blocking along with reduced coverage. In addition, the performance of a user (specially CEUs) gets worse in HetNet environment because of additional inter-tier interference from ubiquitously distributed low powered BSs like FAPs.
The intra-tier and inter-tier interference can be controlled by power control methods \cite{Wang_2015,Chandrasekhar_powercontrol} and/or spectrum allocation methods \cite{Boudreau_2009,Huawei_2005,Junyi_1999,Mahmud_2014,YoungjuKim_2010}. Investigation of power control methods is out of the scope of this paper as our focus is on the performance evaluation of CCUs and CEUs in heterogeneous scenario for different spectrum allocation methods.   Therefore, we have assumed rather a simple power allocation method where each BS-user link is established using equal transmission power irrespective of their distances.

 In literature, spectrum access techniques, like fractional frequency reuse (FFR) \cite{Boudreau_2009,Huawei_2005,Junyi_1999,Mahmud_2014}, soft FFR \cite{SFR,SFR1}, shared spectrum allocation \cite{YoungjuKim_2010}, etc., are investigated to uplift the coverage for CEUs. In FFR \cite{Boudreau_2009,Huawei_2005,Junyi_1999,Mahmud_2014}, the spectrum is divided into cell center and cell edge bands such that the cell center band is accessed with reuse one and cell edge band is accessed with reuse three. 
%  In FFR \cite{Boudreau_2009,Huawei_2005,Junyi_1999,Mahmud_2014}, the spectrum is divided into cell center and cell edge bands. The cell center band is allocated for CCUs with frequency reuse one. However, the cell edge band is further split into three sub-bands which are allocated for CEUs with frequency reuse of three. 
Thus the coverage of CEUs is improved at the loss of spectral efficiency in FFR. 
 On the other hand, the orthogonal cell edge sub-bands are allocated in the neighboring cells for CCUs with reduced power level in soft FFR \cite{SFR,SFR1} which improve the spectral efficiency. FFR and soft FFR are also investigated for multi-tier networks \cite{Hossain2013} wherein small BSs are deployed to operate in sub-bands orthogonal to the cell edge sub-band used in the macro cell. This restricts the inter-tier interference to the CEUs. 
In prior literature, performance analysis of these spectrum allocation techniques are studied in the grid based cellular networks \cite{Mahmud_2014,Giuliano_2008,Chang_2013}. 
% Nonetheless, investigation of spectrum allocation techniques, rendering better cell edge connectivity, has acquired limited attention in the literature  under the paradigm of stochastic network modeling. The coverage probability of a CEU in FFR deployed PPP modeled multi-tier network for the best-effort service is obtained in \cite{Novlan_2011,Novlan_2012,Zhuang2014}. Therein, the user having instantaneous SINR (signal to interference and noise ratio) below a threshold is assumed as a CEU. 
However, investigation of the frequency allocation schemes like FFR, soft FFR, etc. have got limited attention in literature under the paradigm of stochastic network modeling because of \cite{ElSawy}: 1) the difficulty in defining the cell center and cell edge regions due to irregular cell sizes and 2) FFR/soft FFR brings the spatial correlation among the BSs having the same cell edge sub-band which violates the PPP assumption. 
The authors  of \cite{Novlan_2011,Novlan_2012,Zhuang2014} have investigated the performance of FFR/soft FFR methods while extending the analysis of \cite{Andrews_2011} and \cite{Dhillon2012,HSJo_2012} under cruel assumptions. To overcome the first problem, the cell center and cell edge users are categorized based on SINR threshold. Further, to tackle the second problem the worst case is considered wherein each BS randomly accesses one of the cell edge sub-bands. 
However, the assumption of instantaneous SINR based classification causes a user to randomly switch between CCU and CEU. Thus, the traffic flow through the disjoint cell center and cell edge bands becomes coupled which make the evaluation of BS activity and blocking probability intractable. 
On the other hand, though the distance-based classification is suitable for the grid based modeled networks \cite{Mahmud_2014}; this classification is not applicable for PPP modeled networks as the cells are irregular in shapes and sizes.  
In our analysis, we have introduced the parameter $R\in[0,1]$ to define the cell center and cell edge portions. The user is presumed to be in cell center region if the ratio of its distances of serving and dominant interfering MBS is greater than $R$. Otherwise the user is presumed to be in cell edge region.
This helps us to eliminate the first issue. However, the second issue is still unresolved which have restricted us from analyzing the spectrum allocation techniques like FFR wherein partially frequency reuse factor is less than one.

In \cite{YoungjuKim_2010}, a shared spectrum allocation (SSA) for two tier cellular networks is presented wherein femto cells and CCUs are considered to share a portion of the spectrum; and the remaining  portion of the spectrum is protected for CEUs. 
However, authors of \cite{YoungjuKim_2010} have ignored the fading effect and interference from co-channel macro base stations. The problem of BS-correlation (2nd problem discussed above) does not stand for analyzing the SSA using stochastic geometry as each macro cell accesses cell center and cell edge bands with reuse one.  Therefore, in this paper we have considered to investigate SSA along with co-channel spectrum allocation (CSA) using stochastic geometry by employing proposed approach for categorizing CCUs and CEUs using parameter $R$.

Furthermore, interference is basically dependent on the co-channel activity of the BSs that is characterized by the traffic intensity i.e. network load. 
The co-channel activity of MBSs increases with increase of the network load that further results in increased interference level. 
% This implies that overall network performance relatively poor in higher load conditions. %Nevertheless, the CEUs always experience the degraded service as compared to the CCUs because of their lesser intolerance level to the activity of co-channel BSs.
% The fundamental approach to bridge the performance gap between CCU and CEU is interference control.
% Design of a network with reduced co-channel activity of interfering BSs for users in cell edge region compared to the users in cell center region can yield ubiquitously equal coverage.
% The existing investigations of spectrum allocation for a stochastic network in literature are limited to the coverage analysis. 
Nevertheless, the analysis of the above mentioned investigations is limited to best-effort traffic scenario only as all the BSs are assumed to be transmitting all the time. However, the actual coverage exceeds in real-time traffic scenario due to non-simultaneous transmission of BSs. 
In such scenarios, the activity factor, probability that a BSs accesses typical channel, distinguish the density of interfering BSs. The load-aware analysis of wireless networks  using the application stochastic geometry is still at its early stage. The notion of load-awareness in coverage analysis is introduced in \cite{Dhillon_2013}. Therein the activity factor of a BS is assumed to be known a priori. The activity factor of a BS in PPP modeled network for inelastic traffic is derived in \cite{Wei_Bao_2014_NearOptimal} under the consideration of unit bandwidth consumption per service arrival. However, the bandwidth requirement for a service is  region dependent as the achievable transmission rate drops as service location gets closer to the cell boundary.
 In \cite{Praful_BlockingProb}, we have presented a framework to evaluate the activity factor and blocking probability for cellular network. Therein, the cell traffic is modeled using multi-dimensional Markov chain using multi-class service arrival rates which are obtained via solving the nonlinear equations of coverage and activity coupling at different SIR thresholds of modulation and coding schemes.
%Authors of \cite{Dhillon_RandomArrival} characterize the system throughput as a function of service arrival rate within the framework of uniformly distributed service in a single cell of fixed coverage under the consideration of single time-frequency slot allocation per arrival. 
% This assumption causes the analysis to differ from realistic scenario under the consideration of real-time traffic as the smaller achievable transmission rate in cell edge region make CEUs to be bandwidth hungry.  
% Therefore, it is essential to evaluate the network performance while taking region wise variation in bandwidth demand into account.     

The major aspects of cellular network design includes call blocking rate and coverage probability which are basically region dependent. 
Therefore, a comprehensive load-aware performance analysis in relation to a CCU  and a CEU is essential for underlying stochastic real-time traffic. 
Furthermore, the cell/network load is an important aspect for cell center and cell edge spectrum partitioning as it decides the co-channel activities of a typical BS for CCUs and CEUs. For example, setting bigger cell edge band reduces the co-channel activity of BSs that can raise the coverage probability for CEUs.
The SSA  has enough impetus to bridge the gap of performance between CCUs and CEUs as it reserves channels for CEUs with evaded inter-tier interference. However, in a lightly loaded network it is comprehensible that the co-channel spectrum allocation, wherein users of all tiers are allowed to access any channel regardless of being CCU or CEU, may yield better performance. This implies the performance of a spectrum allocation technique is subjective to the network parameters such as base stations density, traffic intensity, rate requirement, etc. Therefore, a comparative analysis of shared and co-channel spectrum allocation is required as well.

%Therefore, to understand the impact of  underlying stochastic real-time traffic on performance of a CCU and a CEU, it becomes essential to incarporate stochastic real-time traffic into the activity factor evaluation.
\subsection{Contributions of the paper}
In this paper, we investigate the performance of SSA and CSA in relation to a CCU and a CEU in a two-tier heterogeneous network under fractional real-time network load conditions. The spatial randomness is incorporated in the analysis by modeling the locations of MBSs and FAPs using independent homogeneous PPPs. The randomness in location of arriving users/services is modeled using space-time PPP (STPPP). 
We propose new criteria to define the cell center and cell edge regions.  A user is defined as CEU if it is having ratio of distances between associated MBS and dominant interfering MBS above a certain fraction; otherwise, the user is defined as CCU. 
% The probability of a user being CCU or CEU depends on the threshold fraction. {Therefore, for given threshold, 
Using this criteria, we can split the macro user arrival process into cell center and cell edge arrival processes to aid evaluation of BS activities and blocking probabilities.

In this paper we explicitly derive the expressions for coverage probability of users (CCUs and CEUs) and activity factors of an MBS under SSA and CSA. The evaluation of activity factors and blocking probabilities facilitated by modeling the cell center traffic and cell edge traffic using two independent one dimensional Markov chains (MCs) for SSA and single two dimensional MC for CSA. 
Further, relating the activity factor with energy spent and the call admission probability with overall transmission rate, we acquire insights of area energy efficiency achieved by an MBS under SSA and CSA deployment. {To the best of our knowledge, the area energy efficiency for real-time traffic under such a paradigm of network modeling is not yet investigated.} The derived framework for the evaluation of coverage and blocking probability of a CCU/CEU along with the evaluation of activity factor of an MBS is validated through extensive simulations. {Through numerical results we demonstrate that a properly configured SSA can realize a network with equal blocking probability for CCUs and CEUs at the cost of little increase of overall blocking probability. Further, we also show that a properly configured SSA can yield better coverage probability  and energy efficiency  compared to that under CSA. This gain in performance is found to be higher for lower/moderate density of FAPs. The performance gain in terms of overall coverage is observed to be decreases with increase in the load.} In the following we enlist the contributions of the paper.
\begin{itemize}
  \item Load aware coverage probabilities are derived for the cell center user and cell edge user in two-tiered network under the consideration of real-time service using the application of stochastic geometry.
  \item The developed framework is also extended for the numerical evaluation of the blocking probabilities for cell center/edge user and the area energy efficiency of the network. 
  \item The performance analysis is studied and compared for shared spectrum allocation and co-channel spectrum allocation. 
  \item The derived results are validated using extensive simulation results. Further, comparative analysis of shared spectrum allocation and co-channel spectrum allocation schemes for coverage, blocking probability and energy efficiency is presented.    
  \item Through numerical results it is demonstrated that properly configured parameter of shared  spectrum allocation can yield a fair blocking probability for cell center user and cell edge users.
 \end{itemize}

The rest of the paper is organized as follows. Section \ref{sec:System-Model} presents the system model and assumptions considered in the analysis. In Section \ref{sec:UserClassification} we discuss the classification of cell center and cell edge traffic.
The analysis of coverage probability and blocking probability under SSA and CSA is presented in Sections \ref{sec:CoverageAnalysis-Shared} and \ref{sec:CoverageAnalysis-Co-Channel} respectively. Next, Section \ref{sec:Area-Energy-Efficiency-Evaluation} presents the area energy efficiency analysis for SSA and CSA techniques.
The numerical results are discussed in Section \ref{sec:Numerical-Results-and-Discussion}. Finally, Section \ref{sec:Conclusion} concludes the paper. 
\section{System Model}
\label{sec:System-Model}
In this paper, we have considered a two-tier cellular network comprised of sets of MBSs and FAPs. Similar to \cite{Dhillon2012}, we model the locations of MBSs and FAPs  using independent homogeneous Poisson point processes \mbox{\small{$\Phi_B$}} and \mbox{\small{$\Phi_F$}} with densities  \mbox{\small{$\lambda_B$}} and \mbox{\small{$\lambda_F$}} units/m$^{2}$ in \mbox{\small{$\mathbb{R}^2$}}, respectively. In the following subsections we categorize the system model in detail and state the assumptions.  
\subsection{Spectrum allocation techniques}
Let  \mbox{\small{$\mathcal{N}$}} be the set of \mbox{\small{$N$}} channels each of bandwidth \mbox{\small{$B$}}.
Fig. \ref{fig:Spectrumallocation} shows the spectrum access policy under SSA technique. In SSA, the \mbox{\small{$p_m$}} fraction of spectrum is reserved for the CCUs and the FAPs; and remaining \mbox{\small{$1-p_m$}} fraction of spectrum is dedicated for the CEUs. The SSA allows the avoidance of inter-tier interference for CEUs which helps in improving the link quality in the cell edge region. On the other hand, under CSA technique, regardless of being cell center or cell edge, the user has access to a channel from the set \mbox{\small{$\mathcal{N}$}}. Moreover, full spectrum access is allowed for FAPs under CSA.

The transmission power and density of inter-tier BSs/APs is required for evaluation of CCUs' and CEUs' performance in heterogeneous scenario.  Therefore, we have considered a simplistic model to accommodate the impact of femto tier interference into the analysis. It is assumed that FAPs are uniformly accessing any one of the channel. Hence, effective density of co-channel FAPs becomes \mbox{\small{$\lambda_F/(Np_m$)}} under SSA and   \mbox{\small{$\lambda_F/N$}} under CSA. Nevertheless, the presented analysis can be easily extended to the scenario wherein complex channel access scheme is employed at the FAP to guarantee QoS by plugging the resultant co-channel FAPs density.
% Further, for simplicity we assumed that FAPs are uniformly accessing any one of the channel. Therefore, effective density of co-channel FAPs becomes \mbox{\small{$\lambda_F/(Np_m$)}} under SSA and   \mbox{\small{$\lambda_F/N$}} under CSA. 
\begin{figure}[htp]
\centering
\includegraphics[trim=7cm 21.7cm 6cm 4cm, width=.23\textwidth]{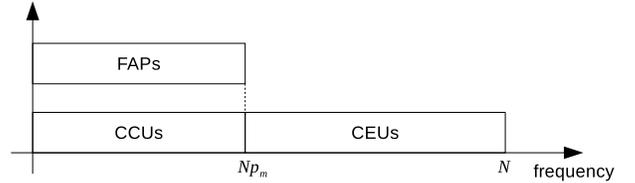} 
\caption{Description of SSA.}
\label{fig:Spectrumallocation}
\end{figure}
\subsection{SIR model}
The transmission power of an MBS and an FAP is denoted by  \mbox{\small{$P_B$}} and  \mbox{\small{$P_F$}}, respectively.  
Macro user is assumed to be associated to the MBS with maximum mean received power. A general power-law path loss model is used. Without loss of generality, a user is assumed to be located at the origin. The location of the serving MBS is denoted by \mbox{\small{$x_0$}}, i.e. \mbox{\small{$x_0=\arg\max_{x_i\in\Phi_B}P_B\|x_i\|^{-\alpha}$}}, where \mbox{\small{$\alpha$}} is the path loss exponent. 
Therefore, the signal-to-interference ratio (SIR) of CCU and CEU, given the associated MBS is at distance $r$ (i.e. \mbox{\small{$r=\|x_0\|$}}), under SSA can be written as\begingroup\makeatletter\def\f@size{8}\check@mathfonts
\begin{equation}
 \gamma_c(r)=\frac{h_0r^{-\alpha}}{I_c(r)+I_f} \text{ ~~~~and~~~~ }\gamma_e(r)=\frac{h_0r^{-\alpha}}{I_e(r)},
\end{equation}\endgroup
respectively. Where  \mbox{\small{$I_c(r)=I_e(r)=\sum_{x_i\in\Phi_B\setminus x_0}h_i\|x_i\|^{-\alpha}$}} represent the co-channel MBSs interference,   \mbox{\small{$I_f=\sum_{x_i\in\Phi_F}h_i\|x_i\|^{-\alpha}\tilde P_F$}} represents co-channel FAPs interference, and \mbox{\small{$\tilde P_F=\frac{P_F}{P_B}$}}. The \mbox{\small{$h_i$'s}} represent fading coefficients and are considered to be i.i.d random variables following exponential distribution with unit mean. 
As the considered network is interference-limited, we ignore the noise in the analysis. However, under CSA the only change required in the above SIR model is to add $I_f$ in the denominator of \mbox{\small{$\gamma_e(r)$}}.  
\subsection{Traffic modeling for CCUs and CEUs}
The macro tier is considered to provide real-time service having rate requirement of \mbox{\small{$R_{\text{th}}$}}.
The call arrival is modeled using space-time PPP (STPPP) \mbox{\small{$\Phi_U$}} with density \mbox{\small{$\lambda_M$}} units/(min$\cdot$m$^{2}$) in \mbox{\small{$\mathbb{R}^3$}} \cite{Praful_BlockingProb}. 
The admitted service stays in the network for an exponential distributed time with mean \mbox{\small{1/$\mu$}} min. 
Thus, the snap shot of admitted users in the network  becomes a PPP with intensity $\frac{\lambda_M}{\mu}(1-\mathcal{B})$ units/m$^2$ in \mbox{\small{$\mathbb{R}^2$}} where $\mathcal{B}$ represents blocking probability \cite{Kleinrock}.
% \subsection{Few definitions}
% In this subsection we provide definitions of some metrics that are used in the paper.
% \begin{itemize}
%  \item \textit{Coverage probability (CovP):} It is the probability that a typical user experiences SIR above threshold \mbox{\small{$\beta$}}. 
%  \item \textit{Activity factor:} The activity factor is defined as the probability that an MBS accesses a typical channel.
%  \item \textit{Blocking probability (BlocP):} It is defined as the probability that a user (call) is blocked from admission into the network.
%  \item \textit{Area energy efficiency ($\eta$):} It is defined as the ratio of average transmission rate of an MBS per unit area to average energy spent by an MBS.
% \end{itemize}
% In following sections we conduct the performance analysis of a CCU and a CEU in two-tier HetNets under SSA and CSA techniques, subsequently.
\section{Classification of CCUs and CEUs Traffic}
\label{sec:UserClassification}
The distances between interfering MBSs and serving MBS are random as their locations are independent. This limits us to determine some fixed distance for classification of a user as cell center or cell edge. Further, referring to Voronoi tessellation of MBSs, the user, having its distances from serving and dominant interfering MBSs relatively closer, appears as flung in the cell edge region. Therefore, we term user as CEU if  \mbox{\small{$\frac{R_m}{R_d}>R$}} otherwise as CCU, where  \mbox{\small{$R_m$}} and \mbox{\small{$R_d$}} are distances of the closest and the second closest points in \mbox{\small{$\Phi_B$}} to the origin, and  \mbox{\small{$R$}} is a predefined fraction.  
% We introduced the notion of user being cell edge or cell center by comparing the ratio of distances of closet and next closet MBSs. 
% The user is considered to be cell edge (center) if $\frac{R_m}{R_d}> R$ ($\frac{R_m}{R_d}\leq R$), where $R$ is fraction. 
The joint distribution of $R_m$ and $R_d$ can be written as \cite{Moltchanov} \begingroup\makeatletter\def\f@size{8}\check@mathfonts
\begin{equation}
f_{R_m,R_d}\left(r_m,r_d\right)=\left(2\pi\lambda_B\right)^2r_mr_d\exp(-\pi\lambda_B r_d^2).
 \label{eq:Joint_Distribution_Of_R1_and_R2}
\end{equation}\endgroup
Using \eqref{eq:Joint_Distribution_Of_R1_and_R2}, the probability of a user being CCU becomes\begingroup\makeatletter\def\f@size{8}\check@mathfonts
\begin{align}
 \mathbb{P}\left[\frac{R_m}{R_d}\leq R\right]&=\int\nolimits_{r_d=0}^{\infty}\int\nolimits_{r_m=0}^{r_dR}f_{R_m,R_d}\left(r_m,r_d\right)dr_mdr_d= R^2.
 \label{eq:ProbBeingCCU}
\end{align}\endgroup
Therefore, the probability of a user being CEU becomes\begingroup\makeatletter\def\f@size{8}\check@mathfonts
\begin{align}
\mathbb{P}\left[\frac{R_m}{R_d}>R\right]=1- \mathbb{P}\left[\frac{R_m}{R_d}\leq R\right]=1-R^2.
\label{eq:ProbBeingCEU}
\end{align}\endgroup
\begin{figure}[htp]
\centering
 \noindent\includegraphics[trim={3cm 8.5cm 3cm 9cm},clip,width=.26\textwidth]{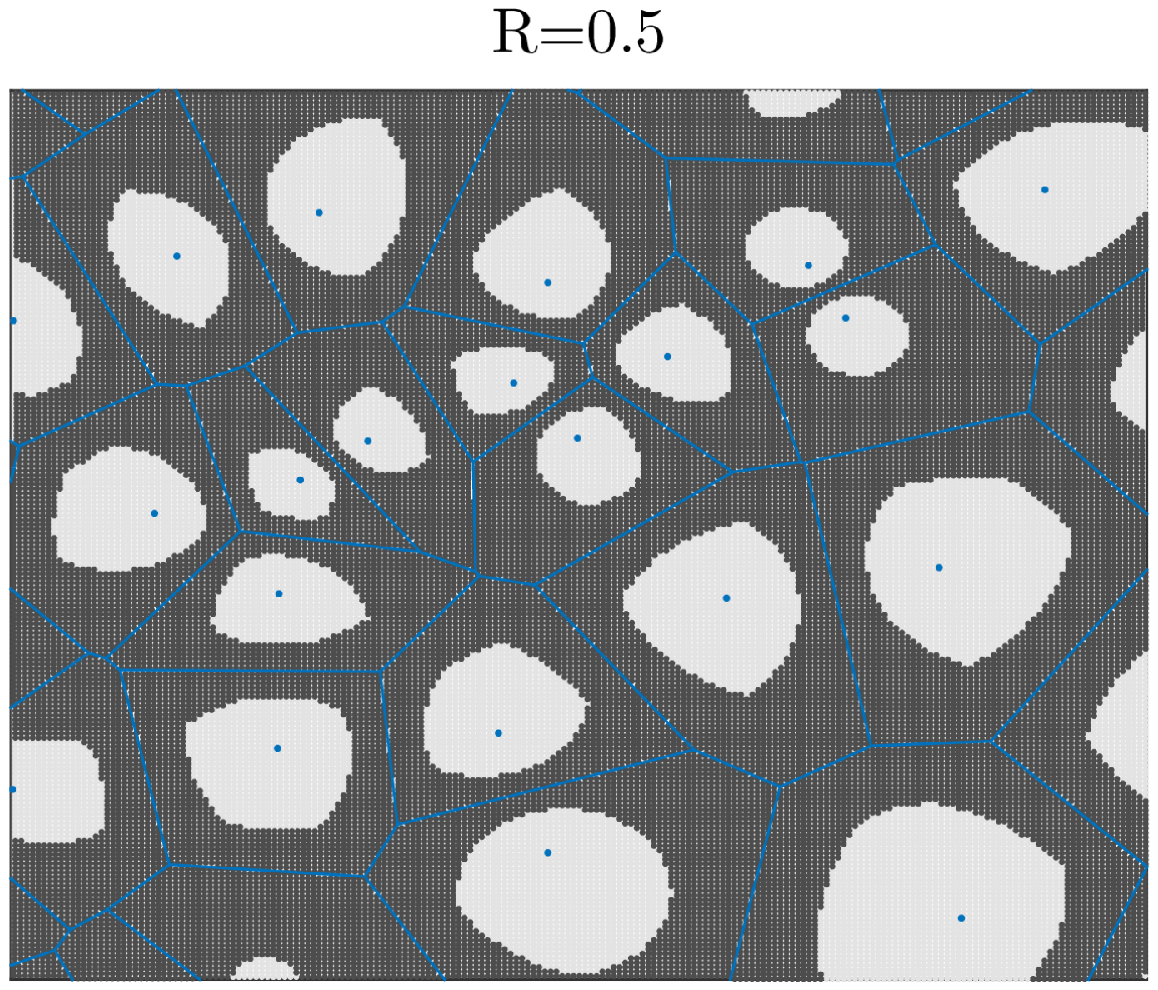} \hspace{-1cm}
 \includegraphics[trim={3cm 8.5cm 3cm 9cm},clip,width=.26\textwidth]{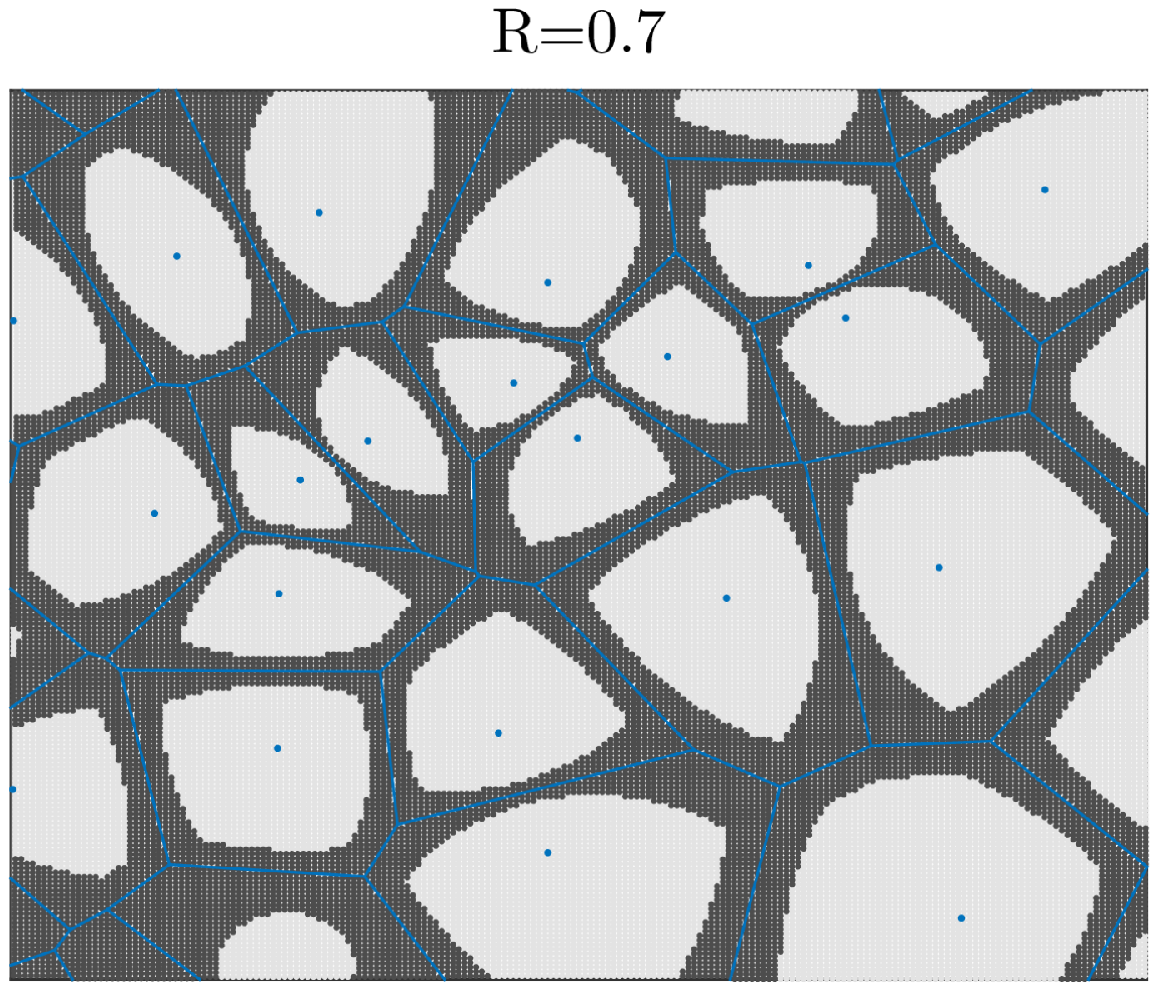} \vspace{-.75cm}
%  \caption{Illustration of cell center region and cell edge  region for \mbox{\small{$R$}} equal to 0.5 and 0.7.} 
 \caption{The Voronoi tessellation of $\Phi_{B}$ where blue dots denote the MBSs, light and dark gray colors represent cell center and cell edge regions.}
\label{fig:Illustration_CUEU}
\end{figure}
Fig. \ref{fig:Illustration_CUEU} illustrate the relevance of the proposed metric for categorizing the users. It can be seen that the cell center/edge region (indicated by light/dark  gray) is monotonically increase/decrease with the increase of threshold \mbox{\small{$R$}}. Therefore, using \eqref{eq:ProbBeingCCU} and \eqref{eq:ProbBeingCEU} we can split the arrival process of macro user of rate \mbox{\small{$\lambda_M$}} into two process, \textit{namely}: 1) cell center user arrival process of rate  \mbox{\small{$\lambda_MR^2$}}, and 2) cell edge user arrival process of rate \mbox{\small{$\lambda_M(1-R^2)$}}. This is possible since the locations of users arrival are independent. Note that users of both processes have service rate of {{$\mu$}}.
\section{Performance Analysis under Shared Spectrum allocation (SSA)}
\label{sec:CoverageAnalysis-Shared}
The coverage probability (CovP) is defined as the probability of a user experiencing the SIR above threshold $\beta$. Denoting interference by $I$ and ignoring noise, the CovP can be expressed using the Laplace transform of $I$ as a function of $\beta$ \cite{Andrews_2011}, i.e. $\mathcal{C}(\beta)=\mathbb{P}(\frac{h_0Pr^{-\alpha}}{I}>\beta)=\mathbb{E}[\exp(I\frac{r^{\alpha}}{P})]$. Furthermore, in random networks the Laplace transform of interference is characterized using the density and transmission power of the co-channel BSs/APs. However, the load dependent transmission of the BSs, like in case of real-time traffic, limits the usage of channel set $\mathcal{N}$ to its fullest. This reduces the overall co-channel activity of the BSs. Assuming BSs access a channel uniformly and independently, the activity factor ($\zeta$) can be interpreted as: 1) probability that a typical BS randomly chooses a typical channel, 2) average fraction of BSs in the network those are co-channel.  Therefore, thinning the BSs density by activity factor determines the co-channel BSs density which can be used to characterize the Laplace transform of interference. It may be noted that the activity factor is dependent on the underlying traffic load and required data rate. However, the achievable rate  depends on the SIR distribution (or CovP) which further decides bandwidth requirement to satisfy the required data rate. This implies that the CovP and the activity factor are coupled together. Consider the example: transmission from a BS, say B1, generates interference to its neighboring BS, say B2, which forces B2 to transmit for a longer time which again interferes back to B1 and make B1 to transmit for even longer time.  The time in this example can be altered with bandwidth in multi-channel scenario. This example clearly depicts the coupling between the CovP and the activity factor.

Further, CovP of a CCU and a CEU is dependent on the spectrum allocation technique as it decides the types and densities of the interfering BSs. Moreover, the spectrum allocation technique along with the CovP determines various attributes of a network such as a user blocking probability (BlocP), BSs transmission rate, etc. Therefore, in this section we conduct the performance analysis for a CCU and a CEU under SSA, which is extended for the same under CSA in Section \ref{sec:CoverageAnalysis-Co-Channel}. In the following subsection, we derive the CovPs for a CCU and a CEU under SSA. In subsequent subsection, a framework is presented for the evaluation of activity factors of an MBS for the cell center and cell edge bands. Therein, we further evaluate the BlocP for CCUs and CEUs.
\subsection{CovP analysis}
In this section, we first evaluate the CovP of a CCU and a CEU for given activity factor of an MBS in the cell center (\mbox{\small{$\zeta_\text{SC}$}}) and the cell edge (\mbox{\small{$\zeta_\text{SE}$}}) bands. In case of fractional activity factor, the effective density of randomly located co-channel MBSs  is thinned by a fraction which is  equal to the activity factor. Different activity factors for cell center band and cell edge band result in different thinning of MBS density in realization of interference processes for a CCU and a CEU.  Therefore, the effective density of interfering MBS for a CCU and a CEU becomes \mbox{\small{$\zeta_\text{SC}\lambda_B$}} and \mbox{\small{$\zeta_\text{SE}\lambda_B$}}, respectively. Confining the FAPs to $p_m$ fraction of spectrum increases the their activity per channel. Therefore, assuming that each FAP accesses one channel, the effective density of FAPs per channel becomes \mbox{\small{$\tilde\lambda_F=\frac{\lambda_F}{p_mN}$}}. Note that the assumption of single channel access for FAP can be relaxed by introducing the activity factor for FAPs as well. 

We first provide the Laplace transform (LT) of intra-tier and inter-tier interference for a CCU in the following lemmas followed by a theorem to evaluate the CovP of a CCU. 
\begin{lemma}
\label{lemma:LT_FAP}
 The LT of femto tier interference to a CCU is \begingroup\makeatletter\def\f@size{8}\check@mathfonts
 \begin{equation}
  \mathcal{L}_{I_f}\left(s\right)=\exp\left(-\delta\pi\tilde\lambda_F(s\tilde P_F)^{\delta}\csc\left[\delta\pi\right]\right).
  \label{eq:LT_FAP}
 \end{equation}\endgroup
 where \mbox{\small{$\delta=\frac{2}{\alpha}$}}.
\end{lemma}
\begin{proof}
 Refer (3.21) of \cite{haenggi2009interference}.
\end{proof}
\begin{lemma}
\label{lemma:LT_MBS_CCU}
 The LT of co-channel MBSs interference for CCU, at distance $r_c$ from its serving MBS, is given by \begingroup\makeatletter\def\f@size{8}\check@mathfonts
 \begin{equation}
  \mathcal{L}_{I_c}\left(s,r_c\right)=\exp\left(-\pi\zeta_\text{SC}\lambda_Bs^{\delta}\int_{\frac{r_c^2}{R^2s^{\delta}}}^\infty \frac{du}{1+u^\frac{1}{\delta}}\right).
  \label{eq:Laplace_MBSInterference_CC}
 \end{equation}\endgroup
\end{lemma}
\begin{proof}
 Referring to Fig. \ref{fig:Example_CUEU}, it is clear that the dominant interfering MBS for a CCU always lies beyond \mbox{\small{$\frac{r_c}{R}$}} when its serving MBS is at $r_c$. Thus, the co-channel interfering MBSs for a CCU situated at $r_c$ from its serving MBS have zero intensity in \mbox{\small{$\mathcal{B}(0,\frac{r_c}{R})$}}\footnote{\mbox{\small{$\mathcal{B}(0,a)$}} represents the ball of radius $a$ centered at origin} and \mbox{\small{$\zeta_\text{SC}\lambda_B$}} in \mbox{\small{$\mathbb{R}^2\setminus\mathcal{B}(0,\frac{r_c}{R})$}}. For rest of proof follow the proof of Theorem 1 of \cite{Andrews_2011}.
\end{proof}
\begin{theorem}
\label{theorem:CCU_shared}
 The CovP of a CCU under SSA is given by\begingroup\makeatletter\def\f@size{8}\check@mathfonts
 \begin{align}
  \mathcal{C}_{\text{\mbox{\tiny{SC}}}}(\beta)&=\left[1+\zeta_\text{SC}R^2\mathcal{H}(\beta,\delta,R)+\pi\delta R^2\frac{\tilde\lambda_F}{\lambda_B}(\beta\tilde{P}_f)^{\delta}\csc\left[\pi\delta\right]\right]^{-1},  
  \label{eq:CovPCC_Shared}
 \end{align}\endgroup
where \mbox{\small{$\tilde\lambda_F=\frac{\lambda_F}{p_mN}$}}, \mbox{\small{$\mathcal{H}(\beta,\delta,R)=\beta^{\delta}\int_{{R^{-2}\beta^{-\delta}}}^\infty \frac{du}{1+u^\frac{1}{\delta}}$}}. For \mbox{\small{$\delta=\frac{1}{2}$}}, we have \mbox{\small{$\mathcal{H}(\beta,\frac{1}{2},R)=\beta^\frac{1}{2}\arctan(R^2\beta^\frac{1}{2})$}}.
\end{theorem}
\begin{proof}
See Appendix \ref{app:theorem_1}.
\end{proof}
\begin{figure}[t]
\centering
\includegraphics[trim=2cm 0.85cm 2cm 0.5cm, width=7cm,height=7cm]{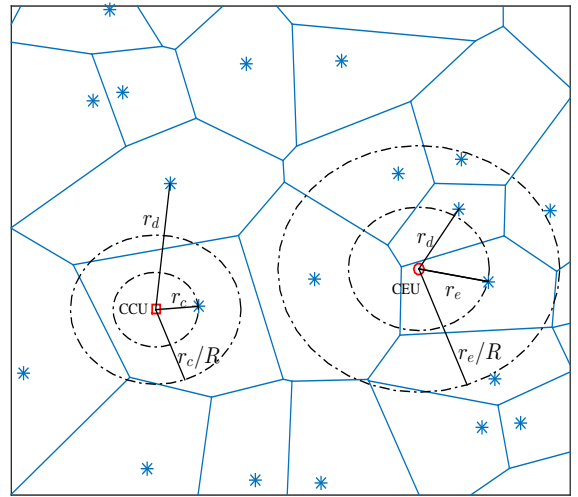} 
%  \caption{Typical example of a CCU (indicated by the square) and a CEU (indicated by the circle).} 
\caption{The Voronoi tessellation of $\Phi_{B}$ where blue star marks denote the MBSs, red square mark denotes CCU, and red circle denote CEU.}
\label{fig:Example_CUEU}
\end{figure}
% Fig. \ref{fig:Example_CUEU} depicts a typical example of a CCU and a CEU.

Now, we provide the LT of intra-tier interference for a CEU in following lemmas followed by a theorem to evaluate the CovP of a CEU.
% Let $r_e$ be the distance between a CEU and its associated MBS. 
From Fig.  \ref{fig:Example_CUEU} can be observed that the serving MBS is at \mbox{\small{$r_e$}} and the dominant MBS is at \mbox{\small{$r_d$}}.  The dominant MBS lies within the ring formed by circles centered at the CEU of radius \mbox{\small{$r_e$}} and \mbox{\small{$\frac{r_e}{R}$}} such that \mbox{\small{$r_e\leq r_d< \frac{r_e}{R}$}}. 
This implies the condition of existence of a dominant MBS within \mbox{\small{$[r_e,\frac{r_e}{R}]$}}. 
Moreover, there is a probability \mbox{\small{$\zeta_\text{SE}$}} (i.e. activity factor) with which an MBS becomes co-channel as a typical MBS access a typical channel with probability \mbox{\small{$\zeta_\text{SE}$}}. Therefore, applying thinning the PPP \mbox{\small{$\Phi_{B}$}} of MBSs by probability \mbox{\small{$\zeta_\text{SE}$}} yields the process of co-channel MBSs which is also a PPP with modified density $\zeta_\text{SE}\lambda_{B}$ . As each MBS chooses a channel independently, there is probability $\zeta_\text{SE}$ with which the dominant MBS persist in thinned process of co-channel MBS and probability $1-\zeta_\text{SE}$ with which the dominant MBS do not persist in thinned process of co-channel MBS.  
% However, the dominant MBS has probability of co-channel transmission equals to $\zeta_{SE}$. 
In following two cases we describe the evaluation of LT of intra-tier interference to CEU.
\newline\textit{Case 1 (dominant MBS does persist in thinned process):} In this case the condition has to be implied about the existence of at least one interfering MBSs inside the annular ring as the dominant MBS is bound to exits within the ring.
Note that in homogeneous PPP the number of nodes in disjoint areas are independent and Poisson distributed. Therefore, we can split the co-channel interfering MBSs in two sets \mbox{\small{$\mathcal{S}_1$}} and \mbox{\small{$\mathcal{S}_2$}}. For given $r_e$, we define the set \mbox{\small{$\mathcal{S}_1(r_e)=\{x\in\Phi_B|~r_e\leq\|x\|\leq \frac{r_e}{R}\}$}} and the set \mbox{\small{$\mathcal{S}_2(r_e)=\{x\in\Phi_B|~\|x\|> \frac{r_e}{R}\}$}}. The definition of CEU implies that the dominant MBS is within \mbox{\small{$\mathcal{S}_1$}} since \mbox{\small{$r_e\leq r_d\leq\frac{r_e}{R}$}}. Therefore, the set \mbox{\small{$\mathcal{S}_1(r_e)$}} must include at least one node. 
Hence, the set of co-channel MBSs for a CEU follows the PPP with zero density in  \mbox{\small{$\mathcal{B}(0,r_e)$}} and \mbox{\small{$\zeta_\text{SE}\lambda_B$}} in  \mbox{\small{$\mathcal{B}(0,\frac{r_e}{R})\setminus\mathcal{B}(0,r_e)$}} (conditioned on minimum one point exist) and \mbox{\small{$\mathbb{R}^2\setminus\mathcal{B}(0,\frac{r_e}{R})$}}.
Let \mbox{\small{$I_{e1}$}} and \mbox{\small{$I_{e2}$}} represents the interference generated from nodes in the set \mbox{\small{${\mathcal{S}_1}(r_e)$}} and \mbox{\small{${\mathcal{S}_2}(r_e)$}}, respectively.
Thus, the resultant interference $I_e$ is the addition of the interference \mbox{\small{$I_{e1}$}} and \mbox{\small{$I_{e2}$}}.  
Let $\mathcal{L}_{I_e}^+(s)$ denote the LT of $I_e$ where the $+$ sign indicate the condition of the dominant MBS is interfering. The $\mathcal{L}_{I_e}^+(s)$ is derived in Lemma \ref{lemma_3}. 
% Further, this segregation of sets of interfering MBSs is useful when the dominant MBS is transmitting. 
\newline\textit{Case 2 (dominant MBS does not persist in thinned process):} In this case the condition of existence of at least one interfering MBS within $\frac{r_e}{R}$ is relaxed as the dominant MBS is not transmitting. Therefore, the set of all interfering MBSs includes the co-channel MBSs beyond distance $r_e$. Hence, the set of co-channel MBSs for a CEU follows the PPP with zero density in  \mbox{\small{$\mathcal{B}(0,r_e)$}} and \mbox{\small{$\zeta_\text{SE}\lambda_B$}} in  \mbox{\small{$\mathbb{R}^2\setminus\mathcal{B}(0,r_e)$}}. The LT of interference $I_e$ in this case is given in Lemma \ref{lemma_4}.

 \begin{figure*}
 \begingroup\makeatletter\def\f@size{7.8}\check@mathfonts
%  \begin{equation}
%   \mathcal{C}_{\text{\mbox{\tiny{SE}}}}(\beta)=\frac{2\pi\lambda_B}{1-R^2}\int_0^{\infty}\frac{\left[\exp\left(-\pi\lambda_B g(\beta,\delta,R)\right)-\exp\left(-cr_e^2\right)\right]}{1-\exp(-cr_e^2)}\exp\left(-\pi\lambda_B h(\beta,\delta,R)\right)\left[\exp\left(-\pi\lambda_Br_e^2\right)-\exp\left(-\pi\lambda_B\frac{r_e^2}{R^2}\right)\right]r_edr_e
%   \label{eq:CovPEC_Dedicate}
%  \end{equation} 
  \begin{align}
  % %  \begin{split}
% % &\mathcal{C}_{\text{\mbox{\tiny{SE}}}}(\beta)=\frac{\zeta_\text{SE}R^2}{1-R^2}\left[\sum\limits_{n=0}^\infty\frac{1}{n\zeta_\text{SE}(1-R^2)+\zeta_\text{SE}R^2g(\beta,\delta,R)+\zeta_\text{SE}R^2h_1(\beta,\delta,R)+R^2}-\frac{1}{n\zeta_\text{SE}(1-R^2)+\zeta_\text{SE}R^2g(\beta,\delta,R)+\zeta_\text{SE}R^2h_1(\beta,\delta,R)+1}-\right.\\
% % &\left.\frac{1}{(n+1)\zeta_\text{SE}(1-R^2)+\zeta_\text{SE}R^2h_1(\beta,\delta,R)+R^2}+\frac{1}{(n+1)\zeta_\text{SE}(1-R^2)+\zeta_\text{SE}R^2h_1(\beta,\delta,R)+1}\right]+\frac{1-\zeta_\text{SE}}{1-R^2}\left[\frac{1}{1+\zeta_\text{SE}h_2(\beta,\delta)}-\frac{R^2}{1+\zeta_\text{SE}R^2h_2(\beta,\delta)}\right]
% % \end{split}\\
   &\mathcal{C}_{\text{\mbox{\tiny{SE}}}}(\beta)=\frac{1-\zeta_\text{SE}}{1-R^2}\sum\limits_{l=0}^{1}\frac{\left(-R^{2}\right)^l}{1+R^{2\cdot l}\zeta_\text{SE}\mathcal{H}(\beta,\delta,1)}+\frac{\zeta_\text{SE}R^2}{1-R^2}\sum\limits_{n=0}^{\infty}\sum\limits_{l=0}^{1}\sum\limits_{k=0}^{1}\frac{\left(-1\right)^{k+l+1}}{\left(n+k\right)\zeta_\text{SE}(1-R^2)+|k-1|\zeta_\text{SE}R^2\mathcal{G}(\beta,\delta,R)+\zeta_\text{SE}R^2\mathcal{H}(\beta,\delta,R)+R^{2\cdot l}},
  \label{eq:CovPEC_Shared1}\tag{10}
  \end{align}
  \text{where}~$\mathcal{G}(\beta,\delta,R)=\mathcal{H}(\beta,\delta,1)-\mathcal{H}(\beta,\delta,R)$. For $\delta=\frac{1}{2}$: $\mathcal{G}(\beta,\frac{1}{2},R)=\beta^\frac{1}{2}[\arctan(\beta^\frac{1}{2})-\arctan(R^2\beta^\frac{1}{2})]$,
\hrule\endgroup 
  \end{figure*} 
\begin{lemma}
\label{lemma_3}
 The LT of co-channel MBSs interference for a CEU, at distance $r_e$ from its serving MBS, conditioned on the transmission of dominant interferer is given by \begingroup\makeatletter\def\f@size{8}\check@mathfonts
 \begin{align}
%  \begin{split}
%    \mathcal{L}_{I_{e1}}\left(s,r\right)&=\frac{1}{1-\exp(-cr^2)}\left[\exp\left(-\frac{cR^2}{1-R^2}s^\delta\int_{{\frac{r^2}{s^{\delta}}}}^{\frac{r^2}{R^2s^{\delta}}} \frac{du}{1+u^\frac{1}{\delta}}\right)\right.\\
%    &~~~~~~~~~~~~~~~~~~~~~~~~~\left.\vphantom{\exp\left(-\frac{cR^2}{1-R^2}s^\delta\int_{{\frac{r^2}{s^{\delta}}}}^{\frac{r^2}{R^2s^{\delta}}} \frac{du}{1+u^\frac{1}{\delta}}\right)}-\exp(-cr^2)\right]
%    \end{split}
&\mathcal{L}_{I_{e}}^+\left(s,r_e\right)=\frac{1}{1-\exp(-cr_e^2)}\exp\left(-\pi\zeta_\text{SE}\lambda_Bs^{\delta}\int_{\frac{r_e^2}{R^2s^{\delta}}}^\infty \frac{du}{1+u^\frac{1}{\delta}}\right)\nonumber\\
&~~~~~~~~~\left[\exp\left(-\pi\zeta_\text{SE}\lambda_Bs^\delta\int_{{\frac{r_e^2}{s^{\delta}}}}^{\frac{r_e^2}{R^2s^{\delta}}} \frac{du}{1+u^\frac{1}{\delta}}\right)-\exp(-cr_e^2)\right],  
\label{eq:Laplace_MBSInterference_CEU_Ie+}
 \end{align}\endgroup
 where \mbox{\small{$c=\pi\zeta_\text{SE}\lambda_B(R^{-2}-1)$}}.
\end{lemma}
\begin{proof}
See Appendix \ref{app:lemma_3}
\end{proof}
% If the dominant interferer is not transmitting, the set of interfering co-channel MBSs for a CEU, at the distance \mbox{\small{$r_e$}} from its serving MBS, follows the PPP with zero density in  \mbox{\small{$\mathcal{B}(0,r_e)$}} and \mbox{\small{$\zeta_\text{SE}\lambda_B$}} in  \mbox{\small{$\mathbb{R}^2\setminus\mathcal{B}(0,r_e)$}}. The following lemma provide the co-channel MBS interference LT in this particular case.   
\begin{lemma}
\label{lemma_4}
 The LT of co-channel MBSs interference for a CEU, at distance \mbox{\small{$r_e$}} from its serving MBS, with the condition that the dominant interferer is not transmitting is given by \begingroup\makeatletter\def\f@size{8}\check@mathfonts
 \begin{equation}
  \mathcal{L}_{I_{e}}\left(s,r_e\right)=\exp\left(-\pi\zeta_\text{SE}\lambda_Bs^{\delta}\int_{\frac{r_e^2}{s^{\delta}}}^\infty \frac{du}{1+u^\frac{1}{\delta}}\right).
  \label{eq:Laplace_MBSInterference_CEU_Ie}
 \end{equation}\endgroup
\end{lemma}
\begin{proof}
Referring to Fig. \ref{fig:Example_CUEU}, it is clear that the dominant interfering MBS for a CEU always lies beyond \mbox{\small{${r_e}$}} when its serving MBS is at $r_e$. Given that the dominant MBS is not transmitting over the same channel, the co-channel MBSs for a CEU situated at $r_e$ from its serving MBS have zero intensity in \mbox{\small{$\mathcal{B}(0,{r_e})$}} and \mbox{\small{$\zeta_\text{SE}\lambda_B$}} in \mbox{\small{$\mathbb{R}^2\setminus\mathcal{B}(0,r_e)$}}. For rest of proof follow the proof of Theorem 1 of \cite{Andrews_2011}.
\end{proof}
\begin{theorem}
\label{theorem:CEU_shared}
 The CovP of a CEU under SSA is given by \eqref{eq:CovPEC_Shared1}
%  \newline\text{where}~$\mathcal{G}(\beta,\delta,R)=\beta^\delta\int\nolimits_{\frac{1}{\beta^\delta}}^{\frac{1}{R^2\beta^\delta}}\frac{dv}{1+v^{\frac{1}{\delta}}}$, \newline$\mathcal{H}_1(\beta,\delta,R)=\beta^\delta\int\nolimits_{\frac{1}{R^2\beta^\delta}}^\infty\frac{dv}{1+v^{\frac{1}{\delta}}}$, and \newline$\mathcal{H}_2(\beta,\delta)=\beta^{\delta}\int_{\frac{1}{\beta^{\delta}}}^\infty \frac{du}{1+u^\frac{1}{\delta}}$. 
%   \newline For $\delta=\frac{1}{2}$: $\mathcal{G}(\beta,\frac{1}{2},R)=\beta^\frac{1}{2}[\arctan(\beta^\frac{1}{2})-\arctan(R^2\beta^\frac{1}{2})]$, \newline$\mathcal{H}_1(\beta,\frac{1}{2},R)=\beta^\frac{1}{2}\arctan(R^2\beta^\frac{1}{2})$, and \newline$\mathcal{H}_2(\beta,\frac{1}{2})=\beta^{\frac{1}{2}}\arctan(\beta^\frac{1}{2})$.
\end{theorem} 
\begin{proof}
  See Appendix \ref{app:theorem_2}
\end{proof}
  Let \mbox{\small{$S_K$}} represent the sums in the second term of \eqref{eq:CovPEC_Shared1} for \mbox{\small{$n=0\dots K$}}. It may be noted that each summand of \mbox{\small{$S_K$}} monotonically decreases with \mbox{\small{$n$}} which implies \mbox{\small{$S_K-S_{K-1} > S_{K+1}-S_K$}}. Therefore, the number of summands can be limited to the value of \mbox{\small{$K$}} where \mbox{\small{$S_{K+1}-S_{K}<\epsilon$}}. It is found that  the \mbox{\small{$S_K$}} converges with insignificant deviation as the value of \mbox{\small{$K$}} attains 10.  This aids to the straightforward numerical evaluation of the CovP for a CUE.
% The infinite number of sums in \eqref{eq:CovPEC_Shared1} can be limited to 10 with an insignificant error. This aids to the straightforward evaluation of the CovP for a CEU.
\subsection{Evaluation of activity factor and BlocP}
\label{sec:TrafficModeling_Shared}
In this section, we derive the activity factor of an MBS for given real-time service. The activity factor is the probability that an MBS accesses a typical channel.
The probability of users being cell center and cell edge is \mbox{\small{$R^2$}} and \mbox{\small{$1-R^2$}} (refer \eqref{eq:ProbBeingCCU} and \eqref{eq:ProbBeingCEU}), respectively. Moreover, the call arrival locations are also independent. Therefore, we can split the call arrival process of rate \mbox{\small{$\lambda_M$}} into two independent arrival processes (cell center and cell edge) with rate \mbox{\small{$\lambda_c=\lambda_MR^2$}} and \mbox{\small{$\lambda_e=\lambda_M(1-R^2)$}}, respectively.  Furthermore, the realized service processes for cell center traffic and cell edge traffic become decoupled, as a dedicated set of channels are allocated. This cause different utilization of cell center and cell edge bands. In other words, an MBS has different activity factors for cell center and cell edge bands. Considering the Shannon capacity, the required number of channels become $n=R_{\text{th}}/(B\log_2(1+\Gamma))$ to meet the service rate $R_\text{th}$ with SIR $\Gamma$. Hence, expected number of required number of channels CCUs ($\bar N_c$) and CEUs ($\bar N_e$) using modulation and coding schemes (MCSs) employed for downlink transmission with the SIR thresholds $\Gamma_i$ for  $i=1\dots T$, can be evaluated as follows
  \begingroup\makeatletter\def\f@size{8}\check@mathfonts
\begin{align}
\setcounter{equation}{10}
 \bar N_c=\sum\limits_{i=1}^{T}\frac{R_{\text{th}}}{B\log_2\left(1+\Gamma_i\right)}&\mathbb{P}_{\text{\mbox{\tiny{SC}}}}(\Gamma_i) \text{~~~and} \label{eq:AVgNumberOfChannelPerMacroService_CCU}\\
 \bar N_e=\sum\limits_{i=1}^{T}\frac{R_{\text{th}}}{B\log_2\left(1+\Gamma_i\right)}&\mathbb{P}_{\text{\mbox{\tiny{SE}}}} (\Gamma_i), \label{eq:AVgNumberOfChannelPerMacroService_CEU}  
\end{align}\endgroup
where\begingroup\makeatletter\def\f@size{7.5}\check@mathfonts
\begin{equation}
 \mathbb{P}_{\text{\mbox{\tiny{SC}}}}(\Gamma_i) =\frac{[\mathcal{C}_{\text{\mbox{\tiny{SC}}}}\left(\Gamma_{i}\right)-\mathcal{C}_{\text{\mbox{\tiny{SC}}}}\left(\Gamma_{i+1}\right)]}{[1-\mathcal{C}_{\text{\mbox{\tiny{SC}}}}\left(\Gamma_{1}\right)]} \text{~and~}  
 \mathbb{P}_{\text{\mbox{\tiny{SE}}}}(\Gamma_i)=\frac{[\mathcal{C}_{\text{\mbox{\tiny{SE}}}}\left(\Gamma_{i}\right)-\mathcal{C}_{\text{\mbox{\tiny{SE}}}}\left(\Gamma_{i+1}\right)]}{[1-\mathcal{C}_{\text{\mbox{\tiny{SE}}}}\left(\Gamma_{1}\right)]}\nonumber
\end{equation}\endgroup
such that \mbox{\small{$\mathbb{P}_{\text{\mbox{\tiny{SC}}}}(\Gamma_i)$}} and \mbox{\small{$\mathbb{P}_{\text{\mbox{\tiny{SE}}}}(\Gamma_i)$}} represents the probability of a CCU and a CEU is served using $i$-th MCS, and \mbox{\small{$\Gamma_{T+1}=\infty$}}.  

The service process is modeled using the STPPP such that 
\begin{itemize}
    \item Number of service arrives in disjoint set are independent.
    \item Number of service arrives in a set $a$ are Poisson distributed with parameter $|a|\lambda_M$.
    \item A service arrival in a set $a$ follows uniform distribution $\frac{1}{|a|}$. 
    \item A service stay for exponentially distributed time with parameter $\frac{1}{\mu}$.  
\end{itemize}
With slight abuse of notation here onwards we use $a$ instead of $|a|$ to represents the area of the set $a$.
Furthermore, the time invariant MBSs are modeled using PPP with density $\lambda_B$. This form the static non-overlapping cells whose shape are defined by voronoi tessellation. This implies that the service process of each cell can be independently modeled. According to STPPP a cell of area $a$ has Poisson service arrival with parameter $a\lambda_M$. Therefore, employing the Little's law the traffic intensity of a cell of area $a$ become $\frac{a\lambda_M}{\mu}$ as the service rate is $\mu$. Hence, the arrival and departure process of services in cell follows memoryless property. The thinning of service arrival with probability $R^2$ will result into cell center and cell edge service process which inherently follows the memoryless property. 
The arrival rates of CCUs' and CEUs' services in a cell of area $a$ are $a\lambda_MR^2$ and $a\lambda_M(1-R^2)$, respectively.  The equivalent number of servers reserved for cell center services and cell edge service at each MBS are $N_c$ and $N_e$ respectively such that $N_c=\lfloor Np_m/\bar N_c\rfloor$ and $N_e=\lfloor N(1-p_m)/\bar N_e\rfloor$. Furthermore, the service are assumed to be admitted with zero waiting time policy.   Therefore, these two independent processes of a macro cell of area $a$ can be modeled using $M/M/N_c/N_c$ and $M/M/N_e/N_e$ queues with arrival rate $a\lambda_c$ and $a\lambda_e$, respectively, and each of departure rate $\mu$. 
%The call arrival rate in a macro cell is dependent on its area as it is defined per unit area. 
%For cell area $a$, the arrival rate of CCU service is \mbox{\small{$a\lambda_c$}} and of CEU service is \mbox{\small{$a\lambda_e$}}. Therefore, these two independent processes of the macro cell of area $a$ can be modeled using \mbox{\small{$M/M/N_c/N_c$}} and \mbox{\small{$M/M/N_e/N_e$}} queues with call arrival rate of \mbox{\small{$a\lambda_c$}} and \mbox{\small{$a\lambda_e$}}, respectively,  and each of departure rate  \mbox{\small{$\mu$}}; where  \mbox{\small{$N_c=\lfloor {Np_m}/{\bar N_c}\rfloor$}} and  \mbox{\small{$N_e=\lfloor {N(1-p_m)}/{\bar N_e}\rfloor$}}. 
Fig. \ref{fig:DecoupledQueue} depicts the decoupled queues. 
% \begin{figure}[h]
% \centering
% \includegraphics[trim=4cm 4.5cm 8cm 5.5cm, width=.5\textwidth]{DecoupledQueue2.eps} 
% \caption{Markov chain for cell center/edge bands under SSA.}
% \label{fig:DecoupledQueue}
% \end{figure}
\begin{figure}[htp]
\centering
\includegraphics[trim=4cm 10cm 4cm 5.5cm, width=.45\textwidth]{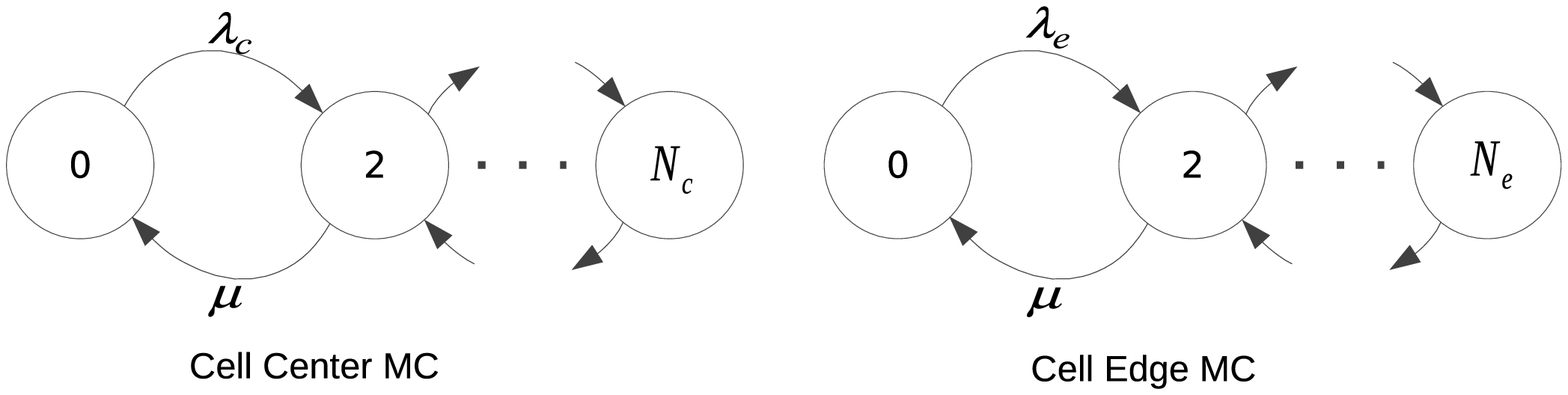} 
\caption{Cell center/edge MC under SSA.}
\label{fig:DecoupledQueue}
\end{figure}
Therefore, the probability that  \mbox{\small{$n$}} servers in the cell center band are occupied is \mbox{\small{$\mathbb{P}\left(n|a\right)=\frac{\left({a\lambda_{c}}/{\mu}\right)^n}{n!}/\sum_{k=0}^{N_c}\frac{\left({a\lambda_{c}}/{\mu}\right)^i}{i!}$}} \cite{Bertsekas:1992}. 
The probability that the MBS uses one server out of $N_c$ cell center servers for a given cell area $a$ is derived in \cite{Wei_Bao_2014_NearOptimal} as follows\begingroup\makeatletter\def\f@size{8}\check@mathfonts
\begin{align}
\label{eq:Channel_Access_Probability_conditional}
\zeta_\text{SC}\left(a\right)&=\frac{1}{N_c}\sum\limits_{n=0}^{N_{c}}n\mathbb{P}\left(n|a\right)\approx\begin{cases}
                                                                                         &\frac{a\lambda_c}{\mu N_c} \mspace{17mu} \text{if}~ \frac{a\lambda_c}{\mu N_c}<1,\\
                                                                                         &1 \mspace{40mu} \text{if}~ \frac{a\lambda_c}{\mu N_c}\geq 1.
                                                                                        \end{cases}
\end{align}\endgroup
The approximation is tighter for larger values of \mbox{\small{$N_c$}}. For larger \mbox{\small{$N$}}, \mbox{\small{$N_c\approx {Np_m}/{\bar N_c}$}}. 
The cell area probability density function is \mbox{\small{$f_A\left(a\right)=\frac{\left({3.5\lambda_B}\right)^{3.5}}{\Gamma\left(3.5\right)}a^{2.5}\exp\left(-3.5a\lambda_B\right)$}} \cite{Singh_2013}. 
% \begingroup\makeatletter\def\f@size{8}\check@mathfonts
% \begin{equation}
% \label{eq:cell_size_PDF}
%  f_A\left(a\right)=\frac{\left({3.5\lambda_B}\right)^{3.5}}{\Gamma\left(3.5\right)}a^{2.5}\exp\left(-3.5a\lambda_B\right)
% \end{equation}\endgroup
Therefore, the activity factor of an MBS for cell center server can be written as we have proven in \cite{Praful_BlockingProb} as follows 
\begingroup\makeatletter\def\f@size{8}\check@mathfonts
 \begin{align}
 \zeta_\text{SC}&=\int_0^\infty \zeta_\text{SC}\left(a\right)f_A\left(a\right)da\nonumber\\
  &=\frac{{\lambda_c\bar N_c}}{3.5{\lambda_B\mu Np_m}\Gamma(3.5)}\gamma\left(4.5,3.5\frac{\lambda_B\mu Np_m}{\lambda_c\bar N_c}\right)\nonumber\\
  &~~~~~~~~~~~~~+\frac{1}{\Gamma(3.5)}\Gamma\left(3.5,3.5\frac{\lambda_B\mu Np_m}{\lambda_c\bar N_c}\right),
  \label{eq:ActivityFactorCCU}
 \end{align}\endgroup
where \mbox{\small{$\gamma(\cdot,\cdot)$}} and \mbox{\small{$\Gamma(\cdot,\cdot)$}} are lower and upper incomplete gamma functions. Similarly, we can derive the activity factor of an MBS for cell edge server as   
\begingroup\makeatletter\def\f@size{8}\check@mathfonts
 \begin{align}
  \zeta_\text{SE}&=\frac{{\lambda_e\bar N_e}}{3.5{\lambda_B\mu N(1-p_m)}\Gamma(3.5)}\gamma\left(4.5,3.5\frac{\lambda_B\mu N(1-p_m)}{\lambda_e\bar N_e}\right)\nonumber\\
  &~~~~~~~~~~~~~~~~~~~+\frac{1}{\Gamma(3.5)}\Gamma\left(3.5,3.5\frac{\lambda_B\mu N(1-p_m)}{\lambda_e\bar N_e}\right).
  \label{eq:ActivityFactorCEU}
 \end{align}\endgroup
 From \eqref{eq:ActivityFactorCCU} and \eqref{eq:ActivityFactorCEU}, it is evident that the activity factors \mbox{\small{$\zeta_\text{SC}$}} and \mbox{\small{$\zeta_\text{SE}$}} are coupled with CovP of CCU and CEU, respectively. {Note that the probability of a server busy is equivalent to the probability of a channel busy. In other words, activity factors of server and channel are equivalent.}  The above expressions are difficult to solve analytically for  \mbox{\small{$\zeta_\text{SC}$}} and \mbox{\small{$\zeta_\text{SE}$}}. Therefore, we evaluate $\zeta_\text{SC}$ and $\zeta_\text{SE}$, recursively, using the bisection method. {The approximated values of \mbox{\small{$\zeta_\text{SC}$}}  and \mbox{\small{$\zeta_\text{SE}$}} are further substituted in \eqref{eq:CovPCC_Shared} and \eqref{eq:CovPEC_Shared1} to evaluate the approximated CovP of a CCU and a CEU, respectively, for a given macro tier fractional load condition.}
 
 Furthermore referring to Fig. \ref{fig:DecoupledQueue}, the BlocP of a CCU and a CEU for a given cell area $a$ can be written \cite{Kleinrock} as\begingroup\makeatletter\def\f@size{8}\check@mathfonts
 \begin{align}
  \mathcal{B}_{\text{\mbox{\tiny{SC}}}}\left(a\right)&=\frac{\left({a\lambda_{c}}/{\mu}\right)^{N_c}}{N_c!}/\sum_{k=0}^{N_c}\frac{\left({a\lambda_{c}}/{\mu}\right)^k}{k!}\text{ and}\label{eq:Blocking-CCU-Shared}\\
  \mathcal{B}_{\text{\mbox{\tiny{SE}}}}\left(a\right)&=\frac{\left({a\lambda_{e}}/{\mu}\right)^{N_e}}{N_e!}/\sum_{k=0}^{N_e}\frac{\left({a\lambda_{e}}/{\mu}\right)^k}{k!},\label{eq:Blocking-CEU-Shared}
  \end{align}\endgroup
respectively. Therefore, the network BlocP becomes
\begingroup\makeatletter\def\f@size{8}\check@mathfonts
\begin{equation}
\label{eq:blocking-probability-shared}
\mathcal{B}_{\text{\mbox{\tiny{S}}}}=\int_0^{\infty}\left(R^2\mathcal{B}_{\text{\mbox{\tiny{SC}}}}\left(a\right)+(1-R^2)\mathcal{B}_{\text{\mbox{\tiny{SE}}}}\left(a\right)\right)f_A\left(a\right)da.
\end{equation}\endgroup
Note that we evaluate the above integral numerically as there exist no closed-form solution to it.
\section{Performance Analysis under Co-channel Spectrum allocation (CSA)}
\label{sec:CoverageAnalysis-Co-Channel}
In this section, we provide the coverage analysis of macro users under CSA for given fractional load conditions. 
In CSA, FAPs and macro users (CCUs/CEUs) are entitled to access any channel from the set \mbox{\small{$\mathcal{N}$}}.  In following subsections, first we provide the coverage analysis for a CCU and a CEU under CSA. Next we present the framework  of modeling cell edge and cell center services using two-dimensional Markov chain for evaluation of activity factor and BlocP. 
\subsection{CovP analysis}
In CSA, the activity of FAPs per channel is reduced by the factor of \mbox{\small{$\frac{1}{p_m}$}} compared to SSA as full bandwidth access is allowed. Therefore, the thinned density of FAPs per channel becomes \mbox{\small{$\tilde\lambda_F=\frac{\lambda_F}{N}$}}. This relaxes the inter-tier interference for the CCUs. However, the CSA exposes the CEUs to the inter-tier interference. Since macro users access channels from the set \mbox{\small{$\mathcal{N}$}}, the MBS has single activity factor \mbox{\small{$\zeta_\text{C}$}}. Therefore, the density of interfering MBSs becomes \mbox{\small{$\tilde\lambda_B=\lambda_B\zeta_\text{C}$}}. The CovPs of a CCU and a CEU are provided in following theorems. 
\begin{corollary}
\label{theorem:CCU_co-channel}
 The CovP of a CCU under CSA is given by\begingroup\makeatletter\def\f@size{8}\check@mathfonts
 \begin{align}
  \mathcal{C}_{\text{\mbox{\tiny{CC}}}}(\beta)&=\left[1+\zeta_\text{C} R^2\mathcal{H}(\beta,\delta,R)+\pi\delta R^2\frac{\tilde\lambda_F}{\lambda_B}(\beta\tilde{P}_f)^{\delta}\csc\left[\pi\delta\right]\right]^{-1},
  \label{eq:CovPCC_Co-channel}
 \end{align}\endgroup
where $\tilde\lambda_F=\frac{\lambda_F}{N}$. % $\mathcal{H}(\beta,\delta,R)=\beta^{\delta}\int_{{R^{-2}\beta^{-\delta}}}^\infty \frac{du}{1+u^\frac{1}{\delta}}$. For $\delta=\frac{1}{2}$, $\mathcal{H}(\beta,\frac{1}{2},R)=\beta^\frac{1}{2}\arctan(R^2\beta^\frac{1}{2})$.
\end{corollary}
\begin{proof}
 Setting \mbox{\small{$p_m=1$}}, \mbox{\small{$\zeta_\text{SC}=\zeta_\text{C}$}}, and following Theorem \ref{theorem:CCU_shared} completes the proof. 
\end{proof}
\begin{corollary}
\label{theorem:CEU_co-channel}
 The CovP of a CEU under CSA is given by \mbox{\small{$\mathcal{C}_{\text{\mbox{\tiny{CE}}}}(\beta)=\mathcal{C}_{\text{\mbox{\tiny{SE}}}}(\beta)$}} (see \eqref{eq:CovPEC_Shared1})
 \text{where}~\mbox{\small{$\zeta_\text{SE}=\zeta_\text{C}$}}, \mbox{\small{$\tilde\lambda_F=\frac{\lambda_F}{N}$}},
 \newline \mbox{\small{$\mathcal{G}(\beta,\delta,R)=\mathcal{H}(\beta,\delta,1)-\mathcal{H}(\beta,\delta,R)$}}, 
 \newline\mbox{\small{$\mathcal{H}(\beta,\delta,R)=\beta^\delta\int\nolimits_{\frac{1}{R^2\beta^\delta}}^\infty\frac{dv}{1+v^{\frac{1}{\delta}}}+\delta\pi\frac{\tilde\lambda_F}{\zeta_\text{C}\lambda_B}(\beta\tilde P_F)^{\delta}\csc[\pi\delta]$}}.
\end{corollary}
\begin{proof}
    The CovP evaluation of a CEU under CSA includes the presence of inter-tier interference in addition to the intra-tier along with the conditional transmission of the dominant interfering MBS.
    Therefore, similar to \eqref{eq:CovP_CC_rc}, the CovP of a CEU at distance \mbox{\small{$r_e$}} from the serving MBS can be written as\begingroup\makeatletter\def\f@size{8}\check@mathfonts
 \begin{align}
  \mathcal{C}_{\text{\mbox{\tiny{CE}}}}(\beta,r_e)&=\zeta_\text{C}\mathcal{L}_{I_{e}}^+(s,r_e)\mathcal{L}_{I_f}\left(s\right)\big |_{s=\beta r_e^\alpha}\nonumber\\
  &+(1-\zeta_\text{C})\mathcal{L}_{I_e}(s,r_e)\mathcal{L}_{I_f}\left(s\right)\big |_{s=\beta r_e^\alpha}.
  \label{eq:CovP_CEU_re_co-channel}
 \end{align}\endgroup
 Therefore, the CovP of a typical CEU can be written as\begingroup\makeatletter\def\f@size{8}\check@mathfonts
 \begin{align}
  \mathcal{C}_{\text{\mbox{\tiny{CE}}}}(\beta)&=\int_{0}^{\infty}\mathcal{C}_{\text{\mbox{\tiny{CE}}}}(\beta,r_e)f_{R_e}(r_e)dr_e\nonumber\\
  \begin{split}
  &=\int_{0}^{\infty}\left[\zeta_\text{C}\mathcal{L}_{I_{e}}^+(\beta r_e^\alpha,r_e)\mathcal{L}_{I_f}\left(s\right)\right.\\
  &\left.~~~~~~~+(1-\zeta_\text{C})\mathcal{L}_{I_e}(\beta r_e^\alpha,r_e)\mathcal{L}_{I_f}\left(s\right)\right]f_{R_e}(r_e)dr_e.
  \end{split}
  \label{eq:CovP_CEU1_co-channel}
 \end{align}\endgroup
 Further, substituting \eqref{eq:LT_FAP}, \eqref{eq:Laplace_MBSInterference_CEU_Ie+}, \eqref{eq:Laplace_MBSInterference_CEU_Ie}, and \eqref{eq:Re_distribution} in \eqref{eq:CovP_CEU1_co-channel} and further following the procedure of Theorem \ref{theorem:CEU_shared} completes the proof.
\end{proof}
\subsection{Evaluation of activity factor and BlocP}
In this section, we discuss the modeling of the macro tier service under CSA. We categorize the traffic in two region-wise classes with call arrival rate of \mbox{\small{$\lambda_c=R^2\lambda_M$}} and \mbox{\small{$\lambda_e=(1-R^2)\lambda_M$}}, as the mean number of channels required by the macro users in cell center region (\mbox{\small{$\bar N_c$}}) and in cell edge region (\mbox{\small{$\bar N_e$}}) are different. Since, CCUs and CEUs access channels from the set \mbox{\small{$\mathcal{N}$}}, we can model the macro cell service using  two dimensional Markov chain (\mbox{\small{$2$-$\mathbb{D}$}} MC). A typical example of \mbox{\small{$2$-$\mathbb{D}$}} MC is shown in Fig. \ref{fig:CoupledQueue} assuming \mbox{\small{$\bar N_c=1$}}, \mbox{\small{$\bar N_e=2$}} and \mbox{\small{$N$}} even. %Note that \mbox{\small{$\bar N_c$}} and \mbox{\small{$\bar N_e$}} are determined as follows:%by substituting CovPs given in Theorem \ref{theorem:CCU_co-channel} and \ref{theorem:CEU_co-channel} in \eqref{eq:AVgNumberOfChannelPerMacroService_CCU} and \eqref{eq:AVgNumberOfChannelPerMacroService_CEU}, respectively.    
Similar to \eqref{eq:AVgNumberOfChannelPerMacroService_CCU} and \eqref{eq:AVgNumberOfChannelPerMacroService_CEU}, the \mbox{\small{$\bar N_c$}} and \mbox{\small{$\bar N_e$}} respectively under CSA can be evaluated as follows
\begingroup\makeatletter\def\f@size{8}\check@mathfonts
\begin{align}
 \bar N_c=\sum\limits_{i=1}^{T}\frac{R_{\text{th}}}{B\log_2\left(1+\Gamma_i\right)}&\mathbb{P}_{\text{\mbox{\tiny{CC}}}}(\Gamma_i) \text{  and} \label{eq:AVgNumberOfChannelPerMacroService_CoChannel_CCU}\\
 \bar N_e=\sum\limits_{i=1}^{T}\frac{R_{\text{th}}}{B\log_2\left(1+\Gamma_i\right)}&\mathbb{P}_{\text{\mbox{\tiny{CE}}}} (\Gamma_i), \label{eq:AVgNumberOfChannelPerMacroService_CoChannel_CEU}  
\end{align}\endgroup
where\begingroup\makeatletter\def\f@size{7.5}\check@mathfonts
\begin{equation}
 \mathbb{P}_{\text{\mbox{\tiny{CC}}}}(\Gamma_i) =\frac{[\mathcal{C}_{\text{\mbox{\tiny{CC}}}}\left(\Gamma_{i}\right)-\mathcal{C}_{\text{\mbox{\tiny{CC}}}}\left(\Gamma_{i+1}\right)]}{[1-\mathcal{C}_{\text{\mbox{\tiny{CC}}}}\left(\Gamma_{1}\right)]} \text{~and~}  
 \mathbb{P}_{\text{\mbox{\tiny{CE}}}}(\Gamma_i)=\frac{[\mathcal{C}_{\text{\mbox{\tiny{CE}}}}\left(\Gamma_{i}\right)-\mathcal{C}_{\text{\mbox{\tiny{CE}}}}\left(\Gamma_{i+1}\right)]}{[1-\mathcal{C}_{\text{\mbox{\tiny{CE}}}}\left(\Gamma_{1}\right)]}.\nonumber
\end{equation}\endgroup

% \begingroup\makeatletter\def\f@size{8}\check@mathfonts
% \begin{align}
%  \bar N_c&=\sum\limits_{i=1}^{T}\frac{R_{\text{th}}}{B\log_2\left(1+\Gamma_i\right)} \frac{[\mathcal{C}_{\text{\mbox{\tiny{CC}}}}\left(\Gamma_{i+1}\right)-\mathcal{C}_{\text{\mbox{\tiny{CC}}}}\left(\Gamma_{i}\right)]}{[1-\mathcal{C}_{\text{\mbox{\tiny{CC}}}}\left(\Gamma_{1}\right)]}\text{~~~~~~and}\label{eq:AVgNumberOfChannelPerMacroService_CCU_co-channel}\\%+\frac{R_{\text{th}}\mathcal{C}_{\text{\mbox{\tiny{CC}}}}\left(\Gamma_{T}\right)}{B\log_2\left(1+\Gamma_T\right)[1-\mathcal{C}_{\text{\mbox{\tiny{CC}}}}\left(\Gamma_{1}\right)]}
%  \bar N_e&=\sum\limits_{i=1}^{T}\frac{R_{\text{th}}}{B\log_2\left(1+\Gamma_i\right)} \frac{[\mathcal{C}_{\text{\mbox{\tiny{CE}}}}\left(\Gamma_{i+1}\right)-\mathcal{C}_{\text{\mbox{\tiny{CE}}}}\left(\Gamma_{i}\right)]}{[1-\mathcal{C}_{\text{\mbox{\tiny{CE}}}}\left(\Gamma_{1}\right)]}%+\frac{R_{\text{th}}\mathcal{C}_{\text{\mbox{\tiny{CE}}}}\left(\Gamma_{T}\right)}{B\log_2\left(1+\Gamma_T\right)[1-\mathcal{C}_{\text{\mbox{\tiny{CE}}}}\left(\Gamma_{1}\right)]}
%  \label{eq:AVgNumberOfChannelPerMacroService_CEU_co-channel}
% \end{align}\endgroup
\begin{figure}[htp]
\centering
\includegraphics[trim=1cm 2.5cm 5cm 1.75cm, width=.5\textwidth]{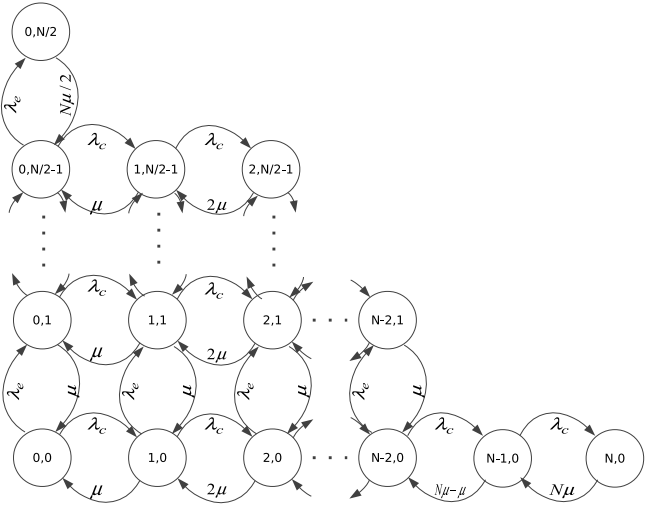} 
\caption{\mbox{\small{$2$-$\mathbb{D}$}} Markov chain under CSA.}
\label{fig:CoupledQueue}
\end{figure}
The activity factor of the cell of area $a$ can be evaluated using the state probabilities of \mbox{\small{$2$-$\mathbb{D}$}} MC. The exact channel access probability and BlocP of 2-\mbox{\small{$\mathbb{D}$}} MC modeled system can be recursively determined using Kaufman-Roberts algorithm \cite{Kaufman,Roberts}.
However, employing Erlang loss model \cite{Karray_2010} we have approximated 2-\mbox{\small{$\mathbb{D}$}} MC with 1-\mbox{\small{$\mathbb{D}$}} MC to bring analytical tractability in the evaluation of activity factor. In Erlang loss model the servers required per service are averaged over the number of servers required per class individually weighted with its traffic intensity. Therefore, the number of servers required per service in the approximated model becomes\begingroup\makeatletter\def\f@size{8}\check@mathfonts
\begin{equation}
 \bar N=\frac{1}{\lambda_M}\left[\lambda_c \bar N_c +\lambda_e \bar N_e\right].
 \label{eq:AVgNumberOfChannelPerMacroService_co-channe}
\end{equation}\endgroup
In this way, the effective call arrival rate in a cell of area $a$ and the effective number of servers of the approximated queue becomes \mbox{\small{$\lambda_Ma=\lambda_ca+\lambda_ea$}} and \mbox{\small{$N_{\text{eff}}=\left\lfloor{N}/{\bar N}\right\rfloor$}}, respectively. 
Similar to \eqref{eq:ActivityFactorCCU} and \eqref{eq:ActivityFactorCEU}, the activity factor of the approximated queue can be written as 
\begingroup\makeatletter\def\f@size{8}\check@mathfonts
 \begin{align}
 \zeta_\text{C}=\frac{{\lambda_M\bar N}}{3.5{\lambda_B\mu N}\Gamma(3.5)}&\gamma\left(4.5,3.5\frac{\lambda_B\mu N}{\lambda_M\bar N}\right)  \nonumber\\
 +\frac{1}{\Gamma(3.5)}&\Gamma\left(3.5,3.5\frac{\lambda_B\mu N}{\lambda_M\bar N}\right). 
  \label{eq:ActivityFactor_co-channel}
 \end{align}\endgroup 
 
It is clear that the CovPs of CCU and CEU (refer Corollary \ref{theorem:CCU_co-channel} and \ref{theorem:CEU_co-channel}) are coupled with the activity factor \eqref{eq:ActivityFactor_co-channel} via \eqref{eq:AVgNumberOfChannelPerMacroService_co-channe}. 
The above expression is difficult to solve analytically for the activity factor \mbox{\small{$\zeta_\text{C}$}}. Therefore, we evaluate \mbox{\small{$\zeta_\text{C}$}} recursively using the bisection method. The approximated value of \mbox{\small{$\zeta_\text{C}$}} is used further to evaluate the approximated CovP of a CCU and a CEU, respectively, for a given macro tier fractional load condition. 

In order to evaluate precise blocking portability we use two dimensional MC (refer Fig. \ref{fig:CoupledQueue}) instead of using Erlang loss model. For given cell area $a$, the traffic intensities of these two classes are \mbox{\small{$\rho_ca=\frac{\lambda_MaR^2}{\mu}$}} and \mbox{\small{$\rho_ea=\frac{\lambda_Ma(1-R^2)}{\mu}$}}. Let \mbox{\small{$\mathbf{n}=[\bar N_c,  \bar N_e]^\top$}} where \mbox{\small{$\bar N_c$}} and \mbox{\small{$\bar N_e$}} is evaluated from \eqref{eq:AVgNumberOfChannelPerMacroService_CoChannel_CCU} and \eqref{eq:AVgNumberOfChannelPerMacroService_CoChannel_CEU} using obtained CovPs of CCU and CEU for given macro tier load. Let \mbox{\small{$\mathcal{S}=\{\mathbf{s}=[s_c,s_e]|\mathbf{s\cdot n}\leq N\}$}} be the possible states of the MC where $s_c$ and $s_e$ denote the number of CCUs and CEUs. For given cell area $a$, the probability of state $\mathbf{s}$ is given by \cite{KeithRoss}\begingroup\makeatletter\def\f@size{8}\check@mathfonts
\begin{align}
 \mathbf{\pi}\left(\mathbf{s}|a\right)&=\frac{1}{G}\frac{(\rho_c a)^{s_c}(\rho_e a)^{s_e}}{s_c!s_e!}, \text{~~where~~} 
 G&=\sum\limits_{\mathbf{s}\in\mathcal{S}}\frac{(\rho_c a)^{s_c}(\rho_e a)^{s_e}}{s_c!s_e!}.
\end{align}\endgroup
Let \mbox{\small{$\mathcal{S}_c$}} (\mbox{\small{$\mathcal{S}_e$}}) be the set of states in which the arriving CCU (CEU) class is admitted, i.e. \mbox{\small{$ \mathcal{S}_c=\left\{\mathbf{s}\in \mathcal{S}:\mathbf{s}\cdot\mathbf{n}\leq N-\bar N_c \right\}$}}  and \mbox{\small{$ \mathcal{S}_e=\left\{\mathbf{s}\in \mathcal{S}:\mathbf{s}\cdot\mathbf{n}\leq N-\bar N_e \right\}$}}. Therefore, the BlocP of CCU and CEU can be written as \begingroup\makeatletter\def\f@size{8}\check@mathfonts
\begin{align}
 \mathcal{B}_{\text{\mbox{\tiny{CC}}}}(a)&=1-\sum\limits_{\mathbf{s}\in\mathcal{S}_c}\mathbf{\pi}\left(\mathbf{s}|a\right) \text{ and }\label{eq:Blocking-CCU-Co-channel}\\
 \mathcal{B}_{\text{\mbox{\tiny{CE}}}}(a)&=1-\sum\limits_{\mathbf{s}\in\mathcal{S}_e}\mathbf{\pi}\left(\mathbf{s}|a\right), \label{eq:Blocking-CEU-Co-channel}
\end{align}\endgroup
respectively, for a given cell area $a$. Therefore, the network BlocP becomes
\begingroup\makeatletter\def\f@size{8}\check@mathfonts
\begin{equation}
\label{eq:blocking-probability-co-channel}
\mathcal{B}_{\text{\mbox{\tiny{C}}}}=\int_0^{\infty}\left(R^2\mathcal{B}_{\text{\mbox{\tiny{CC}}}}\left(a\right)+(1-R^2)\mathcal{B}_{\text{\mbox{\tiny{CE}}}}\left(a\right)\right)f_A\left(a\right)da.
\end{equation}\endgroup
Note that we have to evaluate the above integral numerically as there exist no closed-form solution to it.
\subsection{Extension to orthogonal spectrum allocation}   
   The presented analysis can be directly extended to orthogonal spectrum allocation (OSA) technique. The performance for OSA can be evaluated just by setting number of channel equal to $N(1-p_o)$ and femto density equal to zeros in the expressions derived for CSA where $p_o\in[0,1]$ represent the portion of bandwidth is reserved for femto tier. 
\section{Area Energy Efficiency Evaluation}
\label{sec:Area-Energy-Efficiency-Evaluation}
% The transmission rate can be improved by interference management techniques among which the spectrum allocation is one. The spectrum allocation technique and power consumption are related through the activity of an MBS for real-time services. More interference controlled spectrum allocation yields more power saving. Therefore, it is essential to investigate the impact of spectrum allocation on the energy efficiency. 
The area energy efficiency represent the ratio of average transmission rate of an MBS per unit area to average energy spent by an MBS and it is measured in bps/(joules$\cdot$m$^2$). Therefore, the area energy efficiency of
the cell area ($a$) can be written as:\begingroup\makeatletter\def\f@size{8}\check@mathfonts
\begin{equation}
 \eta(a)=\frac{R_{\text{th}}\mathbb{E}\left[X|a\right]}{aNP_B\zeta},
 \label{eq:Area-Energy-Efficiency}
\end{equation}\endgroup
where \mbox{\small{$X$}} is the number of users in the cell, and \mbox{\small{$\zeta$}} is the activity factor of an MBS.
The above expression shows the relation between area energy efficiency and the traffic intensity. 
% The spectrum allocation techniques define the types and densities of co-channel interferers which further decide the bandwidth consumption. Therefore, the energy expenditure, activity factor dependent, becomes dependent on the spectrum allocation technique. 
In following subsections we evaluate the area energy efficiency for SSA ({{$\eta_{\text{S}}$}}) and CSA ({{$\eta_{\text{C}}$}}) techniques.
\subsection{Evaluation of {{$\eta_{\text{S}}$}}}
\label{sec:Energy-Efficiency-Shared}
 In SSA mode, the disjoint subsets of channel are reserved for cell center region and cell edge region.
 The arrival processes of users in cell center region and in cell edge region are independent. Therefore, the expected number of users admitted into cell becomes the summation of expected number of admitted CCUs and expected number of admitted CEUs. Thus, for \mbox{\small{$M/M/N_c/N_c$}} and \mbox{\small{$M/M/N_e/N_e$}} MCs (refer Fig. \ref{fig:DecoupledQueue}) the expected number of users in a cell of area $a$ becomes \cite{Kleinrock}\begingroup\makeatletter\def\f@size{8}\check@mathfonts
 \begin{equation}
  \mathbb{E}[X|a]=R^2\frac{\lambda_Ma}{\mu}\left[1-\mathcal{B}_{\text{\mbox{\tiny{SC}}}}(a)\right]+(1-R^2)\frac{\lambda_Ma}{\mu}\left[1-\mathcal{B}_{\text{\mbox{\tiny{SE}}}}(a)\right],
 \end{equation}\endgroup
 where \mbox{\small{$\mathcal{B}_{\text{\mbox{\tiny{SC}}}}(a)$}} and \mbox{\small{$\mathcal{B}_{\text{\mbox{\tiny{SE}}}}(a)$}} are the BlocPs of CCUs and CEUs given in \eqref{eq:Blocking-CCU-Shared} and \eqref{eq:Blocking-CEU-Shared}, respectively. Moreover, in the shared mode of spectrum allocation the overall activity of an MBS is \mbox{\small{$\zeta_\text{S}=p_m\zeta_\text{SC}+(1-p_m)\zeta_\text{SE}$}} as \mbox{\small{$\zeta_\text{SC}$}} and \mbox{\small{$\zeta_\text{SE}$}} are the activity factors for \mbox{\small{$p_m$}} and \mbox{\small{$1-p_m$}} fractions of the spectrum, respectively. Therefore, the area energy efficiency for SSA can be evaluated as follows \begingroup\makeatletter\def\f@size{8}\check@mathfonts
 \begin{align}
 &\eta_{\text{S}}=\int_0^\infty\eta_{\text{S}}(a)f_A(a)da\nonumber\\
 \begin{split}
  &=\frac{\lambda_MR_{\text{th}}}{\mu NP_B\zeta_\text{S}}\int_0^\infty\left(R^2\left[1-\mathcal{B}_{\text{\mbox{\tiny{SC}}}}(a)\right]+(1-R^2)\left[1-\mathcal{B}_{\text{\mbox{\tiny{SE}}}}(a)\right]\right)f_A(a)da
 \end{split} \nonumber\\
%  \begin{split}
%   &=\frac{\lambda_MR_{\text{th}}}{\mu NP_B\zeta_\text{S}}\left(R^2\left[1-\int_0^\infty\mathcal{B}_{\text{\mbox{\tiny{SC}}}}(a)f_A(a)da\right]+\right.\\
%   &\left.~~~~~~~~~~~~~~~~(1-R^2)\left[1-\int_0^\infty\mathcal{B}_{\text{\mbox{\tiny{SE}}}}(a)f_A(a)da\right]\right)
%  \end{split} \nonumber\\
 &=\frac{\lambda_MR_{\text{th}}}{\mu NP_B\zeta_\text{S}}\left(R^2\left[1-\mathcal{B}_{\text{\mbox{\tiny{SC}}}}\right]+(1-R^2)\left[1-\mathcal{B}_{\text{\mbox{\tiny{SE}}}}\right]\right).
 \label{eq:Area-Energy-Efficiency-Shared-CellSize-a}
\end{align}\endgroup 
%where \mbox{\small{$(1-P_{\text{out}})=R^2\mathcal{C}_{\text{\mbox{\tiny{SC}}}}(\beta)+(1-R^2)\mathcal{C}_{\text{\mbox{\tiny{SE}}}}(\beta)$}} and \mbox{\small{$\beta$}} is outage threshold (refer Theorem \ref{theorem:CCU_shared} and \ref{theorem:CEU_shared}). 
Above expression clearly shows that the {{$\eta_{\text{S}}$}} is depending on the parameter $p_m$ (i.e number of channel for cell center/edge region which decides BlocPs) of SSA.
\subsection{Evaluation of {{$\eta_{\text{C}}$}}}
\label{sec:Energy-Efficiency-Co-channel}
In CSA mode, the user accesses channels from set \mbox{\small{$\mathcal{N}$}} regardless of its region. Therefore, the BlocP experienced by a CCU/CEU under co-channel mode is different compared to that under SSA and is given in \eqref{eq:Blocking-CCU-Co-channel}/\eqref{eq:Blocking-CEU-Co-channel}. The expected number of users admitted for given cell area $a$  becomes \mbox{\small{$\mathbb{E}[X|a]=\frac{\lambda_Ma}{\mu}\left[1-\mathcal{B}(a)\right]$}} where \mbox{\small{$\mathcal{B}(a)=R^2\mathcal{B}_{\text{\mbox{\tiny{CC}}}}(a)+(1-R^2)\mathcal{B}_{\text{\mbox{\tiny{CE}}}}(a)$}}. 
% Moreover, the activity factor of an MBS under CSA mode (\mbox{\small{$\zeta_\text{C}$}}) is derived for the 2D MC model \eqref{eq:ActivityFactor_co-channel}. 
Therefore, similar to $\eta_{\text{S}}$, we can evaluate $\eta_{\text{C}}$ as follows:\begingroup\makeatletter\def\f@size{8}\check@mathfonts
\begin{equation}
 \eta_{\text{C}}=\frac{\lambda_MR_{\text{th}}\left[1-\mathcal{B}_{\text{\mbox{\tiny{C}}}}\right]}{\mu NP_B\zeta_\text{C}}.
\end{equation}\endgroup 
%where \mbox{\small{$(1-P_{\text{out}})=R^2\mathcal{C}_{\text{\mbox{\tiny{CC}}}}(\beta)+(1-R^2)\mathcal{C}_{\text{\mbox{\tiny{CE}}}}(\beta)$}} (refer Theorem \ref{theorem:CCU_co-channel} and \ref{theorem:CEU_co-channel}), \mbox{\small{$\beta$}} is outage threshold. 
\section{Numerical Results and Discussion}
\label{sec:Numerical-Results-and-Discussion}
In this section, we first validate the developed analytical framework for evaluation of CovP and BlocP of a CCU/CEU along with the activity factor of an MBS through extensive simulations. Next, we present detailed numerical analysis for CovP and BlocP under SSA and CSA techniques. The assessment of CovP and BlocP is carried out with respect to the macro-tier  traffic load (\mbox{\small{$\lambda_M$}}) and the FAP density (\mbox{\small{$\lambda_F$}}). Next, we discuss area energy efficiency aspect of SSA and CSA techniques under different scenarios of FAP interference and macro tier load. We set the distance ratio threshold $R$ equal to $0.707$, which yields equal probability of the user being CCU and CEU. The transmission powers  of an MBS and FAP, i.e. \mbox{\small{$P_B$}} and \mbox{\small{$P_F$}}, are set to  1 Watt and 0.01 Watt per channel, respectively. For numerical analysis, the parameters are considered to be \mbox{\small{$\alpha=4$}}, \mbox{\small{$\beta=1$}}, \mbox{\small{$N=50$}}, \mbox{\small{$\lambda_B=5\times 10^{-6}$}}, \mbox{\small{$\lambda_F=50\lambda_B$}}, \mbox{\small{${R_{\text{th}}}=90$}} kbps, \mbox{\small{$B=180$}} kHz, and \mbox{\small{$\mu=1$}} per min; unless otherwise mentioned.

% For the sake of simplicity we represent legends in the following figures by $\text{X}_{\text{YZ}}$ where replacing $\text{X}$ by $\mathcal{C},~\mathcal{B}$, and $\zeta$ denotes coverage, blocking, and activity factor, respectively;  replacing $\text{Y}$ by $S$ and $C$ denotes SSA and CSA, respectively; replacing $\text{Z}$ by $C$ and $E$ denotes CCU and CEU, respectively. Further, we also use legends $\text{X}_\text{Y}$ to represent overall performance where replacing $\text{X}$ by $\mathcal{C},~\mathcal{B}$, and $\eta$ denotes overall coverage, overall blocking, and area energy efficiency, respectively;  replacing $\text{Y}$ by $S$ and $C$ denotes SSA and CSA, respectively.
\subsection{Validation of presented analytical framework}
\begin{figure}[htp]
\centering
\includegraphics[width=.5\textwidth]{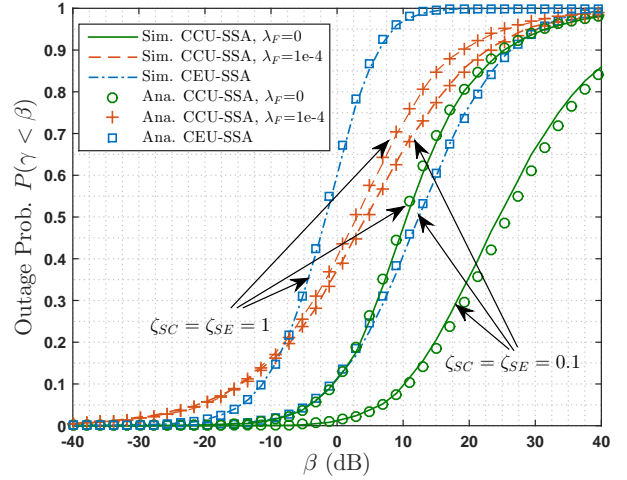} 
\caption{Validation of the outage probability of an CCU and an CEU for SSA.}
\label{fig:OutageValidation}
\end{figure}

\begin{figure}[htp]
\centering
\includegraphics[width=.5\textwidth]{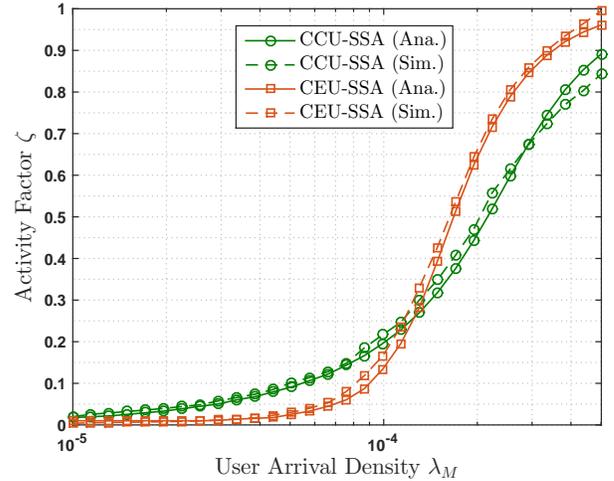} 
\caption{Validation of the activity factors of an MBS for cell center and cell edge bands for SSA.}
\label{fig:ActivityFactorValidation}
\end{figure}
\begin{figure}[htp]
\centering
\includegraphics[width=.5\textwidth]{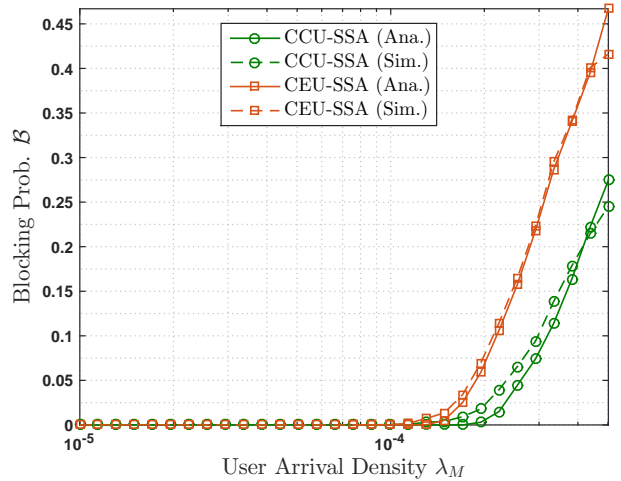} 
\caption{Validation of the blocking probability of an CCU and an CEU under SSA.}
\label{fig:BlockingValidation}
\end{figure}
Note that the outage probability (OutP) can be written as 1-CovP. Fig. \ref{fig:OutageValidation} depicts that the derived OutP expressions of a CCU and a CEU under SSA does match with the results obtained through Monte Carlo simulations for the activity factors equal to 0.1 and 1. It can be seen that the activity factor influences the OutP of a CCU only for smaller FAP density and of a CEU independent of FAP density. However, considering the co-channel interference from  FAPs (usually FAPs have higher density compared to MBSs i.e \mbox{\small{$\lambda_F\gg \lambda_B$}}), the OutP of a CCU becomes more or less independent of the activity factor, as the FAP interference becomes more dominant. Validation of the OutP expressions for CSA is not presented as they are extended from the coverage derivation under SSA. 

Fig. \ref{fig:ActivityFactorValidation} validates the activity factor of an MBS in the cell center (\mbox{\small{$\zeta_\text{SC}$}}) and the cell edge (\mbox{\small{$\zeta_\text{SE}$}}) bands versus user arrival density (\mbox{\small{$\lambda_M$}}) for $p_m=0.4$. It can be seen that the analytically derived activity factors are closely in agreement with the simulation results. Fig. \ref{fig:ActivityFactorValidation} depicts that the activity factors increases with \mbox{\small{$\lambda_M$}}. %However, the rate of increment is different in cell center and cell edge bands. 
The increase of the activity factor is attributed to the following reasons. 1) Increase in  the traffic intensity (i.e. \mbox{\small{$\lambda_M$}}), and 2) Increase in the co-channel interference as increase of interference reduces the achievable transmission rate which further increases the bandwidth requirements. In case of increase of cell center activity factor, only the first reason has a role as increased interference from co-channel MBSs has less impact on achievable rate of CCUs. However, the higher rate of increment in the cell edge activity factor can be observed as both of the above mentioned reasons affect the achievable rate of CEUs which further increases the cell edge bandwidth occupancy.
Fig. \ref{fig:BlockingValidation} validates the blocking probability of an CCU (\mbox{\small{$\mathcal{B}_\text{SC}$}}) and an CEU (\mbox{\small{$\mathcal{B}_\text{SE}$}}) versus user arrival density \mbox{\small{$\lambda_M$}} for $p_m=0.4$. 
\subsection{Numerical analysis of CovP and BlocP}
\begin{figure}[htp]
\centering
\subfigure[]{
    \centering
    \includegraphics[width=.5\textwidth]{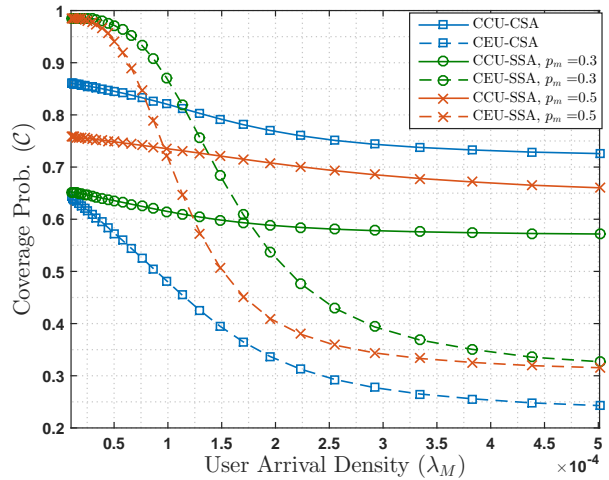}
    \label{fig:CovPvsLamM}
}
\subfigure[]{
    \centering
    \includegraphics[width=.5\textwidth]{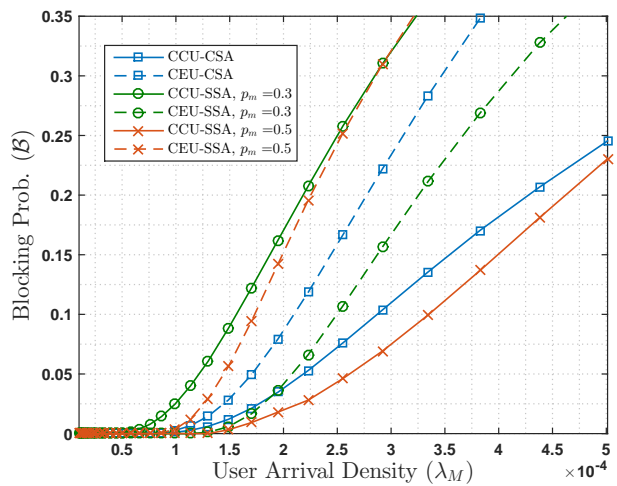} 
    \label{fig:BlockingProbvsLamM} 
}
\caption{Impact of macro tier traffic load \mbox{\small{$\lambda_M$}} on \ref{fig:CovPvsLamM} CovP and \ref{fig:BlockingProbvsLamM} BlocP experienced by CCUs and CEUs for $\lambda_F=100\lambda_B$.} 
\label{fig:CovP_Blocking_Vs_LamM}
\end{figure}
Fig. \ref{fig:CovP_Blocking_Vs_LamM} depicts the impacts of user arrival density (\mbox{\small{$\lambda_M$}}) on CovP and BlocP experienced by CCUs and CEUs. From Fig. \ref{fig:CovPvsLamM} it can be observed that  the CovP drops with an increase in \mbox{\small{$\lambda_M$}}. This is due to the fact that the activity factor of co-channel interfering MBSs increases with the increase of \mbox{\small{$\lambda_M$}}.
It can be observed that the SSA degrades the CovP for CCUs  and improves the CovP for CEUs  compared to CSA. Because in SSA, the density of co-channel FAPs interfering to CCUs increases by a factor ${1}/{p_m}$ and  the inter-tier interference to CEUs is avoided. The degradation and improvement in the CovP of CCUs and CEUs, respectively, can be observed to be increasing with decreasing value of $p_m$.
%However, CEUs experiences higher decrement rate of CovP as compared to CCUs because CEUs are more susceptible to the rising activity of co-channel MBSs. 
% The CovP further depends on the value of $p_m$ which decides the effective density of co-channel MBSs. 
The CovP of CCU (CEU) is monotonically increases (decrease) with $p_m$. Here, we can observe that the CovP of a CCU drops below and the CovP of a CEU rises above those under CSA with drop in $p_m$. %This implies there exist some $p_m$ which yield better CovP for a CEU compared to that in CSA  at the cost of reduced CovP for a CCU.
From Fig. \ref{fig:BlockingProbvsLamM} it is clear that the BlocP experienced by a CEU is relatively higher compared to that of a CCU in CSA as expected. 
% This is mainly due to the fact that a higher bandwidth is required for the services in the cell edge region compared to that in cell center region. 
This phenomenon leads to a network with higher chance of call admission for CCUs compared to CEUs. Under  SSA mode, the BlocP experienced by a CCU (CEU) monotonically decrease (increase) with $p_m$. Fig. \ref{fig:BlockingProbvsLamM} also shows that, the BlocP of a CCU and a CEU is observed to be lesser and higher, respectively, compared to those under CSA for $p_m=0.5$. However, it can be observed that the trend is reversed for a value of $p_m=0.3$.
% with drop in the value of $p_m$ to 0.3 the BlocP of a CCU and a CEU is significantly raised and dropped, respectively, compared to the level of those under CSA. 
This implies that there exits some $p_m$ that can  yield same blocking for a CCU and a CEU for a given \mbox{\small{$\lambda_F$}}. Therefore, in a way SSA allows to realize a network with fair chance of call admission for CCUs and CEUs.

\begin{figure}[htp]
\centering
\subfigure[]{
    \centering
    \includegraphics[width=.5\textwidth]{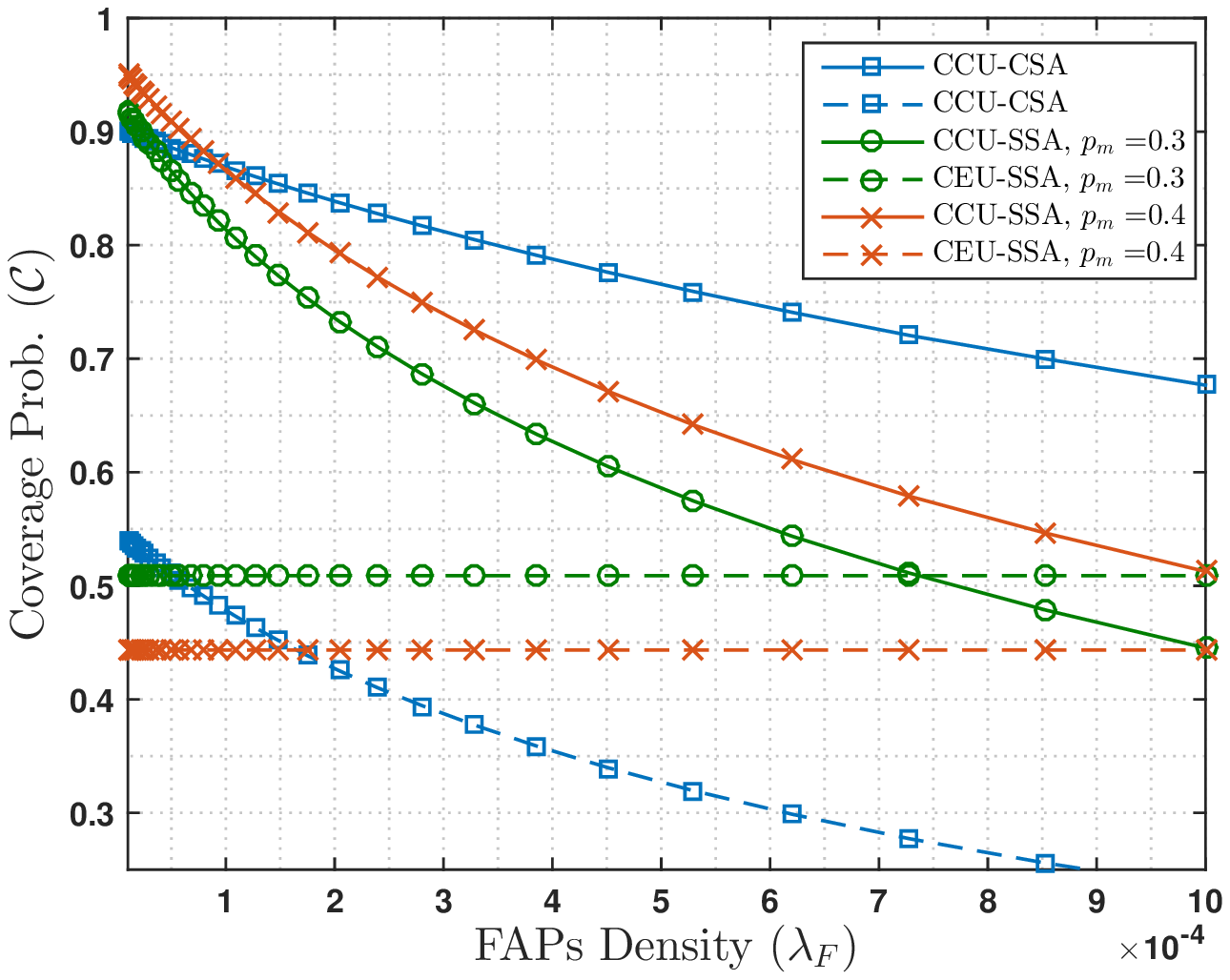} 
    \label{fig:CovPvsLamF}
}
\subfigure[]{
    \centering
    \includegraphics[width=.5\textwidth]{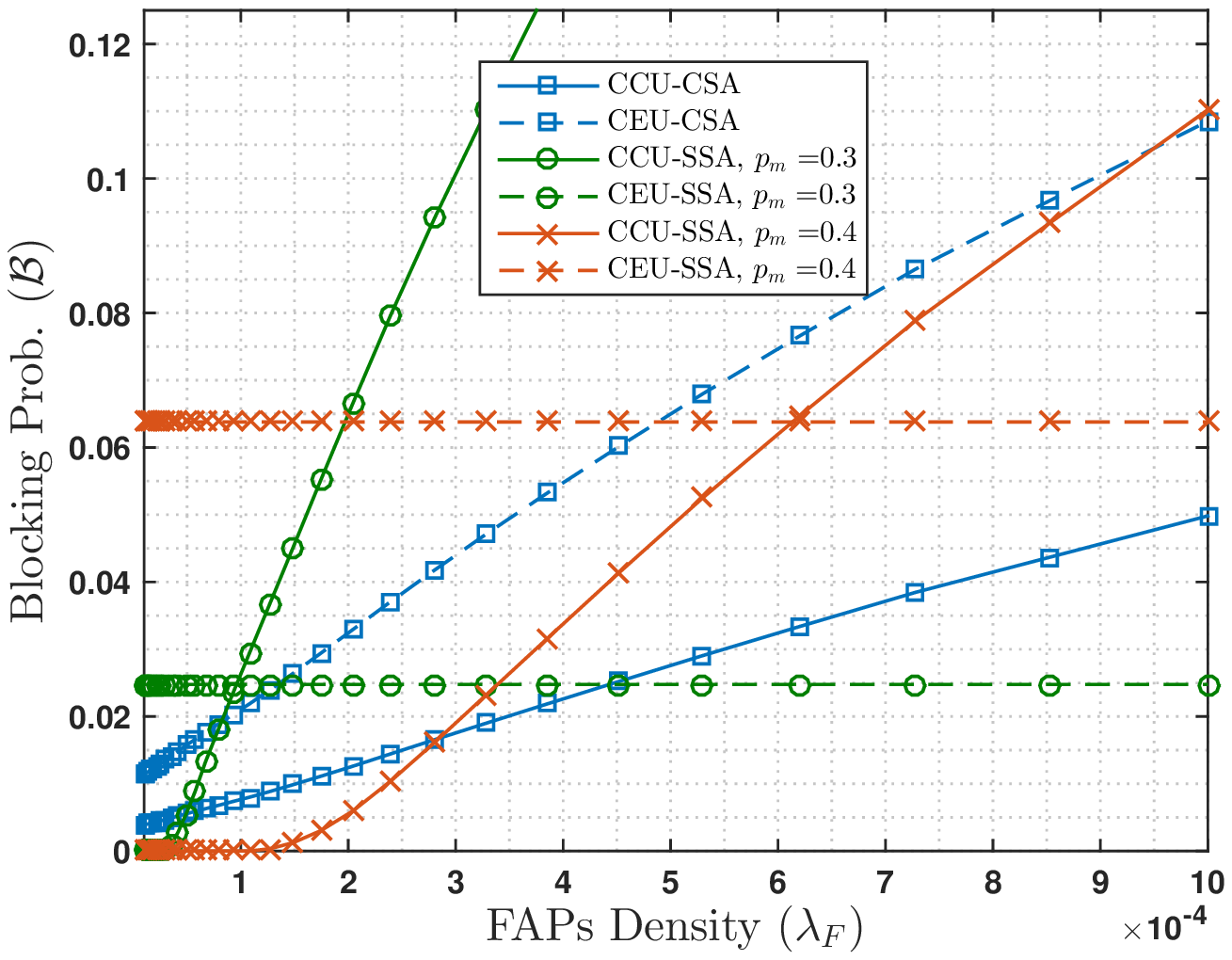} 
    \label{fig:BlockingProbvsLamF} 
}
\caption{Impact of FAP density \mbox{\small{$\lambda_F$}} on \ref{fig:CovPvsLamF} CovP and \ref{fig:BlockingProbvsLamF} BlocP experienced by CCUs and CEUs for \mbox{\small{$\lambda_M=2\times 10^{-4}$}}.} 
\label{fig:CovP_Blocking_Vs_LamF}
\end{figure}
The impacts of FAP density (\mbox{\small{$\lambda_F$}}) on CovP and BlocP experienced by CCUs and CEUs are shown in Fig. \ref{fig:CovP_Blocking_Vs_LamF}. It can be seen that the CovP and the BlocP of both CCUs and CEUs degrades with the increase in \mbox{\small{$\lambda_F$}}. However, the CovP and the BlocP of CEUs under SSA are independent of \mbox{\small{$\lambda_F$}}. Fig. \ref{fig:CovPvsLamF}  depicts that the improved CovP of a CEU under SSA compared to a CSA can be achieved even with a higher value of $p_m$ when FAP density is higher. The smaller value of $p_m$ renders a higher gain in the CovP of a CEU and, moreover, the gain increases with \mbox{\small{$\lambda_F$}}. However, the CovP of a CCU under SSA shifts below its value under CSA.  It can be noted that in the lower range of $\lambda_F$, the SSA yields better CovP to a CCU. This may be due to the fact of dominant interference from co-channel MBSs with reduced activity factor \mbox{\small{$\zeta_{\text{SC}}$}} as compared to the activity factor \mbox{\small{$\zeta_{\text{C}}$}} in CSA.  
% Fig. \ref{fig:BlockingProbvsLamF} depicts that for $p_m=0.3$, SSA yields lower BlocP to CEUs compared to that in CSA in the range where inter-tier interference is dominant (i.e. higher FAP density). 
For given value of $p_m$, Fig. \ref{fig:BlockingProbvsLamF} depicts that there exist a crossover point of FAP density beyond which SSA render lower BlocP to CEUs compared to CSA.  Because under CSA, the inter-tier interference to CEUs increase with increase in FAP density which makes them more and more bandwidth hungry. Further, it can be seen that the crossover points shifts towards right with increase in the value of $p_m$. 
%The crossover point, where the BlocP in CSA goes above its value in SSA, shifts towards left with decreasing $p_m$ as expected.  
Moreover, for lower FAP density and higher value of $p_m$, CCUs experience lower BlocP in SSA compared to CSA as in this scenario the inter-tier interference is insignificant and co-channel MSBs' cell center activity  is also lower. In SSA, BlocP of CCUs increase rapidly as  the effective density of co-channel FAPs raises by a factor $1/p_m$.
The CCUs get lesser band to access and receive severe inter-tier interference as the value of $p_m$ gets smaller.  
The small value of $p_m$ causes higher rate of increase in BlocP of a CCU with respect to the FAP density. Therefore, the $p_m$ must be lower bounded to limit the BlocP for a CCU while improving the blocking for a CEU. 
%From the figure, it can be seen that reduced gap between BlocP of a CCU and a CEU under SSA is realized for $p_m=0.3$ in the lower range of $\lambda_F$ and $p_m=0.4$ in the higher range of $\lambda_F$. 

\begin{figure}[htp]
\centering
\subfigure[]{
    \centering
    \includegraphics[width=.5\textwidth]{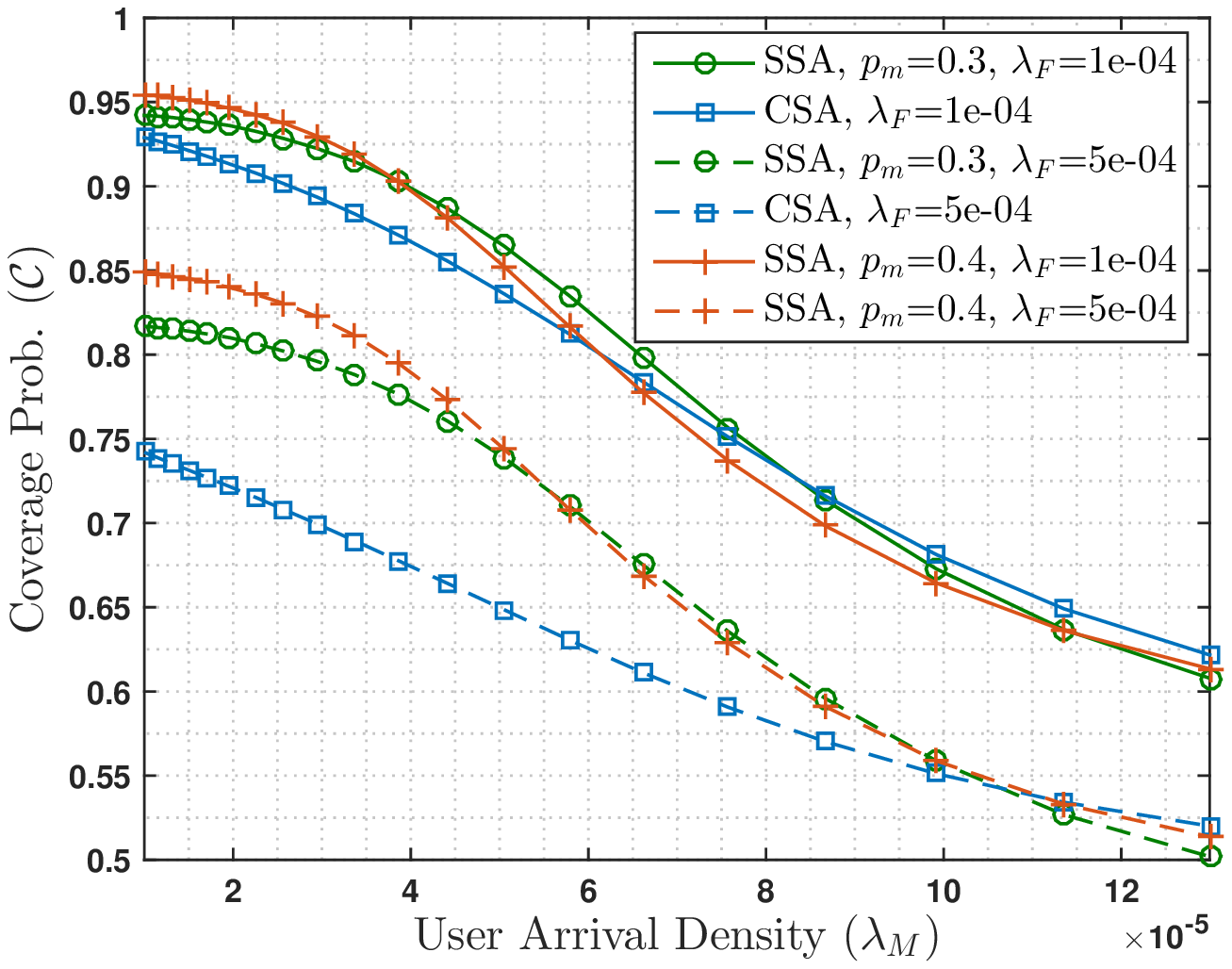} 
    \label{fig:OCovPvsLamM}
}
\subfigure[]{
    \centering
    \includegraphics[width=.5\textwidth]{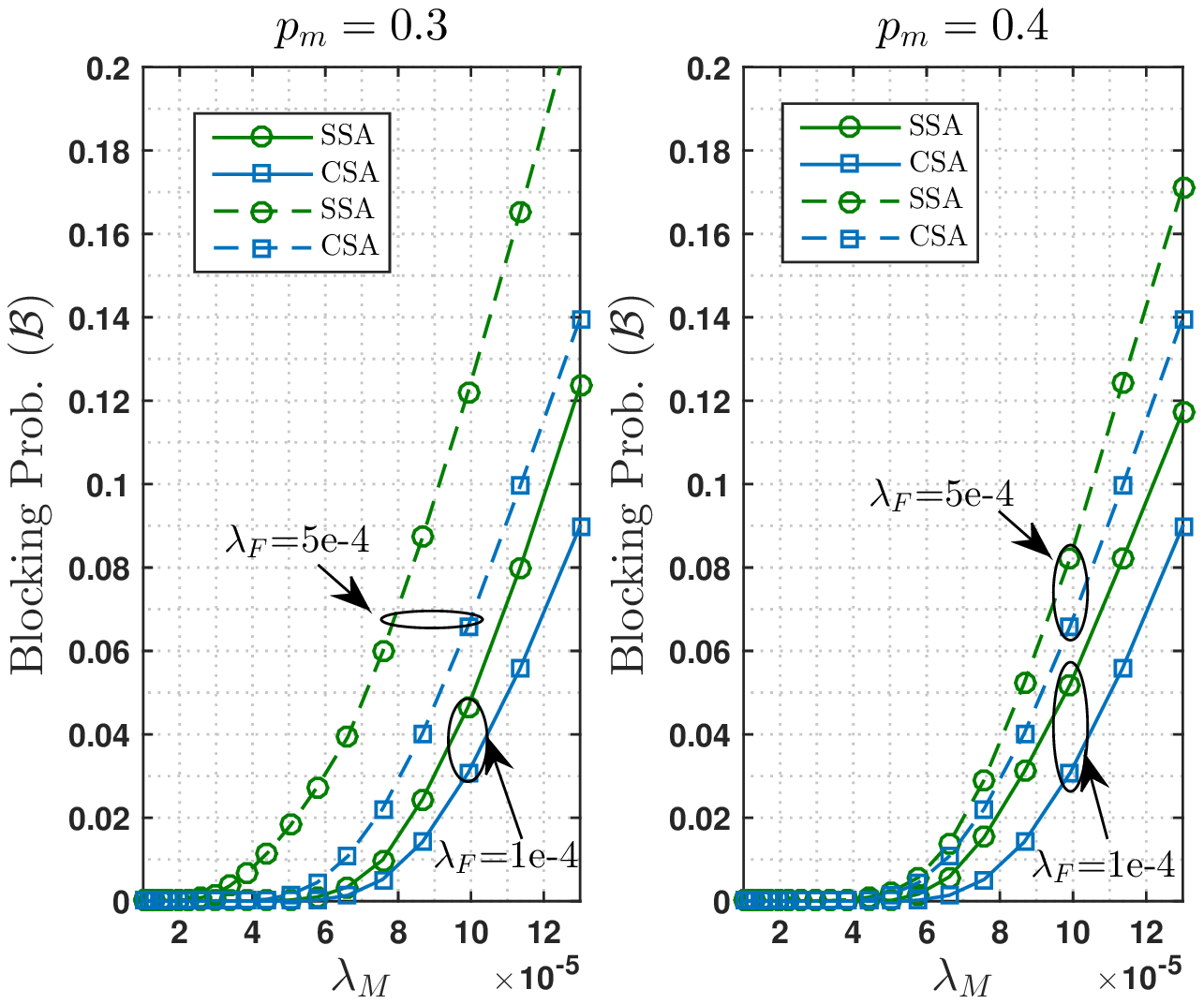} 
    \label{fig:OBlockingProbvsLamM} 
}
\caption{Impact of load \mbox{\small{$\lambda_M$}} on \ref{fig:OCovPvsLamM} overall CovP and \ref{fig:OBlockingProbvsLamM} overall BlocP experienced by a macro user for $R_{\text{th}}=180$ kbps.} 
\label{fig:Overall_CovP_Blocking_Vs_LamM}
\end{figure}
The overall CovP and BlocP of a macro user are depicted in Fig. \ref{fig:Overall_CovP_Blocking_Vs_LamM} as a function of traffic load (\mbox{\small{$\lambda_M$}}) for higher and lower densities of FAPs. From Fig. \ref{fig:OCovPvsLamM}, it is clear  that SSA yields better CovP for a macro user as compared to CSA independent of the scenario. The improvement is significantly higher in the lower values of \mbox{\small{$\lambda_M$}} and higher values of \mbox{\small{$\lambda_F$}}. Further, from Fig. \ref{fig:OBlockingProbvsLamM} it can be seen that SSA yield higher overall BlocP compared to the CSA. However, the gap between the BlocPs under SSA and CSA is dependent on the parameter $p_m$ and the density of FAPs. The BlocP of a macro user under SSA is relatively closer to that under CSA for the lower density of FAPs.  
%It can be seen that in this scenario SSA yield higher (lower) BlocP compared to CSA for $p_m$ = 0.3 (0.4). 
However, for higher density of FAPs the reduced gap can be observed for $p_m=0.4$.
This implies that employment of SSA (with suitable choice of $p_m$ as per scenario) can yield better CovP along with a small increment in BlocP as compared to CSA. Nevertheless, it should be noted that the values of $p_m$ that yields region-wise fair BlocP  and better overall CovP   along with little degraded overall BlocP are not guaranteed to be the same.
\subsection{Numerical analysis of area energy efficiency ($\eta$)}
\begin{figure}[htp]
\centering
\includegraphics[width=.5\textwidth]{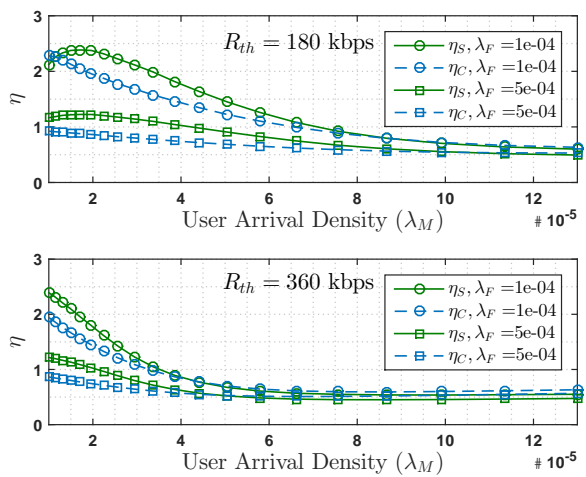} 
\caption{Energy Efficiency ($\eta$) versus traffic load for $p_m=0.3$.}
\label{fig:EnergyEfficiency}
\end{figure}
Fig. \ref{fig:EnergyEfficiency} depicts the area energy efficiency (\mbox{\small{$\eta$}}) as a function of macro tier real-time traffic load (\mbox{\small{$\lambda_M$}}) for rate requirement equal to 180 kbps and 360 kbps. It can be seen that, the values of \mbox{\small{$\eta_{\text{C}}$}} and \mbox{\small{$\eta_{\text{S}}$}}  decrease with increase in the value of \mbox{\small{$\lambda_M$}}. This is due to sudden increase of activity factor with \mbox{\small{$\lambda_M$}} which decides the rate of energy spent. Moreover, the transmission rate increases linearly with \mbox{\small{$\lambda_M$}} till point of having insignificant BlocP. It can be further observe that the \mbox{\small{$\eta_{\text{C}}$}} and \mbox{\small{$\eta_{\text{S}}$}} both saturates at higher \mbox{\small{$\lambda_M$}}. However, it can be seen that SSA provide better area energy efficiency as compared to CSA, i.e. \mbox{\small{$\eta_{\text{S}}>\eta_{\text{C}}$}}, since the activity of an MBS is observed to be lower under SSA as compared to CSA. 

% \begin{figure*}[!h]
% \centering
% \subfigure[]{
%     \centering
%     \includegraphics[width=.45\textwidth]{matlab/CovP_Vs_LamM_LamF5E4_Rth5E1_v2.eps} 
%     \label{fig:CovPvsLamM}
% }
% \subfigure[]{
%     \centering
%     \includegraphics[width=.45\textwidth]{matlab/BlockingProb_Vs_LamM_LamF5E4_Rth5E1_v2.eps} 
%     \label{fig:BlockingProbvsLamM} 
% }
% \subfigure[]{
%     \centering
%     \includegraphics[width=.45\textwidth]{matlab/CovP_Vs_LamF_LamM3.5E4_Rth5E1_v2.eps} 
%     \label{fig:CovPvsLamF}
% }
% \subfigure[]{
%     \centering
%     \includegraphics[width=.45\textwidth]{matlab/BlockingProb_Vs_LamF_LamM3.5E4_Rth5E1_v2.eps} 
%     \label{fig:BlockingProbvsLamF} 
% }
% \caption{Impact of macro tier traffic load \mbox{\small{$\lambda_M$}} on \ref{fig:CovPvsLamM} CovP and \ref{fig:BlockingProbvsLamM} BlocP experienced by CCUs and CEUs. Impact of FAP density \mbox{\small{$\lambda_F$}} on \ref{fig:CovPvsLamF} CovP and \ref{fig:BlockingProbvsLamF} BlocP experienced by CCUs and CEUs.} 
% \label{fig:CovP_Blocking_Vs_LamF} 
% \label{fig:CovP_Blocking_Vs_LamM}
% \end{figure*}
\section{Conclusion}
\label{sec:Conclusion}
Previous works characterizing the interference, coverage, transmission capacity under consideration of best effort traffic in multi-tier heterogeneous networks (HetNets) modeled using independent homogeneous Poisson point processes (PPPs) exclusively focused on a generic user.  In this paper, we extend the analysis explicitly for a cell center/edge user (CCU/CEU) for real-time traffic in a random two-tier HetNet. We have developed a framework for load-aware performance analysis of such networks that includes analysis of blocking probability and area energy efficiency besides the coverage analysis. Identifying that the spectrum allocation technique has an impact on coverage probability, we considered two spectrum allocation techniques for the performance analysis namely: 1) shared spectrum allocation (SSA), and 2) co-channel spectrum allocation (CSA). Our study unveils that a macro user (CCU/CEU) experience better coverage in an instance of lower traffic load as the effective activity of a macro base station (MBS) is smaller. It is observed that the CSA deployment imposes higher blocking probability on CEUs as compared to CCUs. Through numerical results, it is demonstrated that the SSA helps in realizing a fair chance of call admission for CCUs and CEUs by configuring its parameter $p_m$ for a given traffic load. 
Moreover, the SSA can be configured to obtain a better overall coverage with little increment in the overall blocking probability as compared to CSA for lower/moderate  traffic load (depending upon FAP density). Our numerical analysis further reveals that SSA provides higher area energy efficiency compared to CSA. Therefore, the presented numerical analysis demonstrate that the SSA with properly adjusted parameter $p_m$ should be the choice over CSA for deployment in a random two-tier HetNet.

\appendices
\section{Proof of Theorem \ref{theorem:CCU_shared}}
\label{app:theorem_1}
% Assuming the activity factor of an MBS in cell center band is \mbox{\small{$\zeta_\text{SC}$}} and a FAP selects a channel uniformly among $Np_m$, the thinned density of interfering MBSs and FAPs becomes \mbox{\small{$\zeta_\text{SC}\lambda_F$}} and \mbox{\small{$\frac{\lambda_F}{p_mN}$}}. Therefore, 
The CovP of CCU at distance $r_c$ from serving MBS can be written as\begingroup\makeatletter\def\f@size{8}\check@mathfonts
 \begin{align}
  \mathcal{C}_{\text{\mbox{\tiny{SC}}}}(\beta,r_c)&=\mathbb{P}\left({h_0r_c^{-\alpha}}>\beta\left(I_c(r_c)+I_f\right)\right)\nonumber\\
  &\eqa\mathbb{E}\left[\exp\left(\beta r_c^{\alpha}\left(I_c(r_c)+I_f\right)\right)\right]\nonumber\\  
  &=\mathcal{L}_{I_c}(s,r_c)\mathcal{L}_{I_f}(s)\big |_{s=\beta r_c^\alpha}.
  \label{eq:CovP_CC_rc}
 \end{align}\endgroup
 Step (a) is directly followed as \mbox{\small{$h_0$}} is unit mean exponential random variable. Let \mbox{\small{$R_c$}} and \mbox{\small{$R_d$}} be the random variables representing distances of a CCU from the associated MBS and dominant interfering MBS. Therefore, the probability that $R_c$ is greater than $r_c$ becomes\begingroup\makeatletter\def\f@size{8}\check@mathfonts
\begin{align}
 F_{R_c}(r_c)&=\mathbb{P}\left[R_c>r_c|R_c\leq R\cdot R_d\right]=\frac{\mathbb{P}\left[R_c>r_c,R_c\leq R\cdot R_d\right]}{\mathbb{P}\left[R_c\leq R\cdot R_d\right]}\nonumber\\
&\eqa\frac{1}{R^2}\int\nolimits_{r_c}^{\infty}\int\nolimits_{\frac{r_c}{R}}^{\infty}f_{R_c,R_d}\left(r_c,r_d\right)dr_ddr_c\nonumber\\
% &=\frac{\left(2\pi\lambda_B\right)^2}{R^2}\int\nolimits_{r_c}^{\infty}r_1\int\nolimits_{\frac{r_1}{R}}^{\infty}r_d\exp\left(-\pi\lambda_B r_d^2\right)dr_ddr_1\\
&\eqb\exp\left(-\pi\lambda_B \frac{r_c^2}{R^2}\right).
\end{align}\endgroup
Step (a) directly follows using \eqref{eq:ProbBeingCCU} where \mbox{\small{$R_m=R_c$}}. Step (b) is derived through substitution of \eqref{eq:Joint_Distribution_Of_R1_and_R2}. 
Therefore, probability density function of $R_c$ becomes\begingroup\makeatletter\def\f@size{8}\check@mathfonts
\begin{align}
 f_{R_c}(r_c)&=\frac{d}{dr_c}\left[1-F_{R_c}(r_c)\right]=2\pi\lambda_B\frac{ r_c}{R^2}\exp\left(-\pi\lambda_B\frac{r_c^2}{R^2}\right).
 \label{eq:Rc_distribution}
\end{align}\endgroup 
Therefore, the CovP of a typical CCU can be written as\begingroup\makeatletter\def\f@size{8}\check@mathfonts
\begin{align}
\mathcal{C}_{\text{\mbox{\tiny{SC}}}}(\beta)&=\int_0^\infty\mathcal{C}_{\text{\mbox{\tiny{SC}}}}(\beta,r_c)f_{R_c}(r_c)dr_c\nonumber\\
 &\eqa\int_0^\infty\mathcal{L}_{I_c}(\beta r_c^\alpha,r_c)\mathcal{L}_{I_f}(\beta r_c^\alpha)f_{R_c}(r_c)dr_c\nonumber\\ 
\begin{split}
 &\eqb\frac{2\pi\lambda_B}{R^2}\int_0^{\infty}r_c \exp\left(-\pi\lambda_B r_c^2\left[R^{-2}+\zeta_\text{SC}\mathcal{H}(\beta,\alpha)+\vphantom{\frac{\tilde\lambda_F}{\lambda_B}(\beta\tilde P_F)^{\delta}\csc[\pi\delta]}\right.\right.\\
 &~~\left.\left.~~~~~~~~~~~~~~~~~~~~~~~~~~\delta\pi\frac{\tilde\lambda_F}{\lambda_B}(\beta\tilde P_F)^{\delta}\csc[\pi\delta]\right]\right)dr_c.
\end{split} 
\label{eq:CovPCC_Shared1}
\end{align}\endgroup
Step (a) is followed thorough the substitution of \eqref{eq:CovP_CC_rc}. Step (b) is directly followed through the substitution of \eqref{eq:LT_FAP}, \eqref{eq:Laplace_MBSInterference_CC}, and \eqref{eq:Rc_distribution} where \mbox{\small{$\mathcal{H}(\beta,\delta,R)=\beta^{\delta}\int_{{R^{-2}\beta^{-\delta}}}^\infty \frac{du}{1+u^\frac{1}{\delta}}$}}. Further, solving the integral in \eqref{eq:CovPCC_Shared1} completes the proof.
\section{Proof of Lemma \ref{lemma_3}}
\label{app:lemma_3}
The LT of interference \mbox{\small{$I_e=I_{e1}+I_{e2}$}} can be written as\begingroup\makeatletter\def\f@size{8}\check@mathfonts
\begin{equation}
 \mathcal{L}_{I_e}^+(s,r_e)=\mathcal{L}_{I_{e1}}^{+}(s,r_e)\mathcal{L}_{I_{e2}}(s,r_e).
  \label{eq:LT_Ie_conditioned}
\end{equation}\endgroup
Following Lemma \ref{lemma:LT_MBS_CCU} the LT of \mbox{\small{$I_{e2}$}}  can be written as follows\begingroup\makeatletter\def\f@size{8}\check@mathfonts
\begin{equation}
  \mathcal{L}_{I_{e2}}\left(s,r_e\right)=\exp\left(-\pi\zeta_\text{SE}\lambda_Bs^{\delta}\int_{\frac{r_e^2}{R^2s^{\delta}}}^\infty \frac{du}{1+u^\frac{1}{\delta}}\right).
  \label{eq:LT_Ie2}
 \end{equation}\endgroup
 In PPP of intensity $\lambda$, the number of points in the area $A$ follows Poisson distribution with mean $\lambda A$ and each point is uniformly distributed within the area $A$.
Therefore, for given $r_e$, the number of point in \mbox{\small{$\mathcal{S}_1(r_e)$}} are Poisson random variable with mean \mbox{\small{$cr_e^2=\pi\zeta_\text{SE}\lambda_B r_e^2[R^{-2}-1]$}} and each point is independently distributed as \begingroup\makeatletter\def\f@size{8}\check@mathfonts
\begin{equation}
 f(x)=\begin{cases}
  \frac{1}{\pi r_e^2(R^{-2}-1)}~~~~~\text{for}~r_e\leq \|x\|<\frac{r_e}{R},\\
  0~~~~~~~~~~~~~~~~~\text{otherwise}.
 \end{cases}
 \label{eq:Uniform_Pdf}
\end{equation}\endgroup
Since \mbox{\small{$\mathcal{S}_1(r_e)$}} must contain at least one node, the conditional distribution of number of points in \mbox{\small{$\mathcal{S}_1(r_e)$}} becomes\begingroup\makeatletter\def\f@size{8}\check@mathfonts
\begin{equation}
 \mathbb{P}(K=k|K\geq 1)=\frac{\mathbb{P}(K=k)}{1-\exp(-cr_e^2)}, ~~~~\text{for}~k=1,2,\dots
\end{equation}\endgroup
Therefore, the moment generating function (mgf) of \mbox{\small{$K$}} given \mbox{\small{$K\geq 1$}} becomes\begingroup\makeatletter\def\f@size{8}\check@mathfonts
\begin{equation}
 \mathbb{E}[z^k|k\geq 1]=\frac{1}{1-\exp(-cr_e^2)}\left[\mathbb{E}[z^k]-\exp(-cr_e^2)\right].
 \label{eq:Conditional_mgf_Poi1}
\end{equation}\endgroup
Substituting the mgf of a Poisson random variable of mean $cr_e^2$, i.e. \mbox{\small{$\mathbb{E}[z^k]=\exp(-cr_e^2[1-z])$}}, in  \eqref{eq:Conditional_mgf_Poi1} we can write\begingroup\makeatletter\def\f@size{8}\check@mathfonts
\begin{equation}
 \mathbb{E}[z^k|k\geq 1]=\frac{1}{1-\exp(-cr_e^2)}\left[\exp(-cr_e^2[1-z])-\exp(-cr_e^2)\right].
 \label{eq:Conditional_mgf_Poi}
\end{equation}\endgroup
 For given $r_e$, the LT of \mbox{\small{$I_{e1}$}} can be written as \begingroup\makeatletter\def\f@size{8}\check@mathfonts
 \begin{align}
  &\mathcal{L}_{I_{e1}}^+\left(s,r_e\right)=\mathbb{E}\left[\exp\left(-s\sum\nolimits_{x_i\in\mathcal{S}_1(r_e)}h_i\|x_i\|^{-\alpha}\right)\right]\nonumber\\
  &=\mathbb{E}\left[\prod\nolimits_{x_i\in\mathcal{S}_1(r_e)}\mathcal{L}_h\left(s\|x_i\|^{-\alpha}\right)\right]\nonumber\\
  &\eqa\mathbb{E}_k\left[\left[\int\limits_{r_e<\|x\|\leq\frac{r_e}{R}}\mathcal{L}_h\left(s\|x_i\|^{-\alpha}\right)f(x)dx\right]^{k}\bigg |k\geq 1\right]\nonumber\\
  \begin{split}
  &\eqb\frac{1}{1-e^{-cr_e^2}}\left[\exp{\left(-cr_e^2\int\limits_{r_e<\|x\|\leq\frac{r_e}{R}}\left[1-\mathcal{L}_h\left(s\|x_i\|^{-\alpha}\right)\right]f(x)dx\right)} \right.\\
  &~~~~~~~~~~~~~~~\left.\vphantom{\exp{\left(-cr_e^2\int\limits_{r_e<\|x\|\leq\frac{r_e}{R}}\left[1-\mathcal{L}_h\left(s\|x_i\|^{-\alpha}\right)\right]f(x)dx\right)}}-e^{-cr_e^2}\right] \nonumber\\
  &\eqc\frac{1}{1-e^{-cr_e^2}}\left[\exp{\left(-cr_e^2\int\limits_{r_e<\|x\|\leq\frac{r_e}{R}}\frac{f(x)dx}{1+\frac{1}{s}\|x\|^\alpha}\right)}-e^{-cr_e^2}\right]
  \end{split}\nonumber\\
  &\eqd\frac{1}{1-e^{-cr_e^2}}\left[\exp{\left(-\frac{cR^2}{\pi(1-R^2)}\int\limits_{r_e<\|x\|\leq\frac{r_e}{R}}\frac{dx}{1+\frac{1}{s}\|x\|^\alpha}\right)}-e^{-cr_e^2}\right]\nonumber\\
  &\eqe\frac{1}{1-e^{-cr_e^2}}\left[{\exp\left(-\frac{cR^2}{(1-R^2)}s^\delta\int_{{\frac{r_e^2}{s^{\delta}}}}^{\frac{r_e^2}{R^2s^{\delta}}} \frac{du}{1+u^\frac{1}{\delta}}\right)}-e^{-cr_e^2}\right].   
  \label{eq:LT_Ie1}
 \end{align}\endgroup
Step (a) directly follows as each point is independently and identically distributed. 
Step (b) follows using the \eqref{eq:Conditional_mgf_Poi}. Step (c) is obtain through the substitution of LT of the exponential distribution, i.e. {{$\mathcal{L}_h(s)=\frac{1}{1+s}$}}. Substituting {{$f(x)$}} from \eqref{eq:Uniform_Pdf} yields Step (d).
Next,converting from Cartesian to polar coordinates yields Step (e).  

Further, substituting \eqref{eq:LT_Ie2} and \eqref{eq:LT_Ie1} in \eqref{eq:LT_Ie_conditioned} completes the proof.
\section{Proof of Theorem \ref{theorem:CEU_shared}}
\label{app:theorem_2}
The CovP of a CEU is dependent on the activity of the dominant interfering MBS, as it imposes the condition of being in \mbox{\small{$(r_e,\frac{r_e}{R})$}} when the serving MBS is at distance \mbox{\small{$r_e$}}. Therefore, similar to \eqref{eq:CovP_CC_rc}, the CovP of a CEU at distance \mbox{\small{$r_e$}} from the serving MBS can be written as\begingroup\makeatletter\def\f@size{8}\check@mathfonts
 \begin{align}
  \mathcal{C}_{\text{\mbox{\tiny{SE}}}}(\beta,r_e)&=\zeta_\text{SE}\mathcal{L}_{I_{e}}^+(s,r_e)+(1-\zeta_\text{SE})\mathcal{L}_{I_e}(s,r_e)\big |_{s=\beta r_e^\alpha}.
  \label{eq:CovP_CEU_re}
 \end{align}\endgroup
 Now, let \mbox{\small{$R_e$}} be the distance of a CEU from the associated MBS.
Therefore, the probability that \mbox{\small{$R_e$}} is greater than \mbox{\small{$r_e$}} becomes \begingroup\makeatletter\def\f@size{8}\check@mathfonts
\begin{align}
 F_{R_e}(r_e)&=\mathbb{P}\left[R_e>r_e|R_e> R\cdot R_d\right]\nonumber\\
 &=\frac{\mathbb{P}\left[R_e>r_e,R_e> R\cdot R_d\right]}{\mathbb{P}\left[R_e> R\cdot R_d\right]}\nonumber\\
&\eqa\frac{1}{1-R^2}\int\nolimits_{r_e}^{\infty}\int\nolimits_{r_e}^{\frac{r_e}{R}}f_{R_e,R_d}\left(r_e,r_d\right)dr_ddr_e\nonumber\\
% &=\frac{\left(2\pi\lambda_B\right)^2}{1-R^2}\int\nolimits_{r_e}^{\infty}\int\nolimits_{r_e}^{\frac{r_e}{R}}r_er_d\exp\left(-\pi\lambda_B r_d^2\right)dr_ddr_e\\
&\eqb\frac{1}{1-R^2}\left[\exp\left(-\pi\lambda_Br_e^2\right)-R^2\exp\left(-\pi\lambda_B \frac{r_e^2}{R^2}\right)\right].
\end{align}\endgroup
Step (a) directly follows using \eqref{eq:ProbBeingCEU} where \mbox{\small{$R_m=R_e$}} and the fact of \mbox{\small{$R_e$}} and \mbox{\small{$R_d$}} are independent. Step (b) is derived through substitution of the joint density function of \mbox{\small{$R_e$}} and \mbox{\small{$R_d$}}.
Therefore, probability density function of \mbox{\small{$R_e$}} becomes\begingroup\makeatletter\def\f@size{8}\check@mathfonts
\begin{align}
 f_{R_e}(r_e)&=\frac{d}{dr_e}\left[1-F_{R_e}(r_e)\right]\nonumber\\
 &=2\pi\lambda_B\frac{ r_e}{1-R^2}\left[\exp\left(-\pi\lambda_Br_e^2\right)-\exp\left(-\pi\lambda_B \frac{r_e^2}{R^2}\right)\right].
 \label{eq:Re_distribution}
\end{align}\endgroup
Therefore, the CovP of a typical CEU can be written as\begingroup\makeatletter\def\f@size{8}\check@mathfonts
 \begin{align}
  &\mathcal{C}_{\text{\mbox{\tiny{SE}}}}(\beta)=\int_{0}^{\infty}\mathcal{C}_{\text{\mbox{\tiny{SE}}}}(\beta,r_e)f_{R_e}(r_e)dr_e\nonumber\\
  &=\int_{0}^{\infty}\left[\zeta_\text{SE}\mathcal{L}_{I_{e}}^+(\beta r_e^\alpha,r_e)+(1-\zeta_\text{SE})\mathcal{L}_{I_e}(\beta r_e^\alpha,r_e)\right]f_{R_e}(r_e)dr_e.
  \label{eq:CovP_CEU1}
 \end{align}\endgroup
 Further, substituting \eqref{eq:Laplace_MBSInterference_CEU_Ie+}, \eqref{eq:Laplace_MBSInterference_CEU_Ie}, and \eqref{eq:Re_distribution} in \eqref{eq:CovP_CEU1} yields \eqref{eq:CovPEC_Shared} (given at the top of next page).
 \begin{figure*}[t]
 \begingroup\makeatletter\def\f@size{8}\check@mathfonts
  \begin{align}
&\mathcal{C}_{\text{\mbox{\tiny{SE}}}}(\beta)=(1-\zeta_\text{SE})\frac{2\pi\lambda_B}{1-R^2}\int_0^{\infty}\exp\left(-\pi\zeta_\text{SE}\lambda_Br_e^2\mathcal{H}(\beta,\delta,1)\right)\left[\exp\left(-\pi\lambda_Br_e^2\right)-\exp\left(-\pi\lambda_B\frac{r_e^2}{R^2}\right)\right]r_edr_e+
  \label{eq:CovPEC_Shared}\\
  &\zeta_\text{SE}\frac{2\pi\lambda_B}{1-R^2}\int_0^{\infty}\frac{\left[\exp\left(-\pi\zeta_\text{SE}\lambda_Br_e^2 \mathcal{G}(\beta,\delta,R)\right)-\exp\left(-\pi\zeta_\text{SE}\lambda_B(R^{-2}-1)r_e^2\right)\right]}{1-\exp(-\pi\zeta_\text{SE}\lambda_B(R^{-2}-1)r_e^2)}\exp\left(-\pi\zeta_\text{SE}\lambda_Br_e^2 \mathcal{H}(\beta,\delta,R)\right)\left[\exp\left(-\pi\lambda_Br_e^2\right)-\exp\left(-\pi\lambda_B\frac{r_e^2}{R^2}\right)\right]r_edr_e,\nonumber
 \end{align}
  \text{where}~$\mathcal{G}(\beta,\delta,R)=\beta^\delta\int\nolimits_{\frac{1}{\beta^\delta}}^{\frac{1}{R^2\beta^\delta}}\frac{dv}{1+v^{\frac{1}{\delta}}}=\mathcal{H}(\beta,\delta,1)-\mathcal{H}(\beta,\delta,R)$.%, $\mathcal{H}_1(\beta,\delta,R)=\beta^\delta\int\nolimits_{\frac{1}{R^2\beta^\delta}}^\infty\frac{dv}{1+v^{\frac{1}{\delta}}}$, and $\mathcal{H}_2(\beta,\delta)=\beta^{\delta}\int_{\frac{1}{\beta^{\delta}}}^\infty \frac{du}{1+u^\frac{1}{\delta}}$. 
  %\newline For $\delta=\frac{1}{2}$: $g(\beta,\frac{1}{2},R)=\beta^\frac{1}{2}[\arctan(\beta^\frac{1}{2})-\arctan(R^2\beta^\frac{1}{2})]$, $h_1(\beta,\frac{1}{2},R)=\beta^\frac{1}{2}\arctan(R^2\beta^\frac{1}{2})$, and $h_2(\beta,\frac{1}{2})=\beta^{\frac{1}{2}}\arctan(\beta^\frac{1}{2})$. 
  \hrule \endgroup
 \end{figure*}
Using Maclaurin series \mbox{\small{$\frac{1}{1-x}=\sum_{n=0}^\infty x^n$}}, we can write \mbox{\small{$[1-\exp(-x)]^{-1}=\sum_{n=0}^{\infty}\exp(-nx)$}} for \mbox{\small{$x\geq 0$}}. Substituting this series in \eqref{eq:CovPEC_Shared} and further rearranging the solution of integral yields  \eqref{eq:CovPEC_Shared1}. This completes the proof.
% \bibliographystyle{IEEEtran}
% \bibliography{TwoTierPerformanceAnalysis}{}

\begin{thebibliography}{1} 
 \bibitem{Wyner}
  A. D. Wyner, “Shannon-theoretic approach to a gaussian cellular multiple-access channel,” IEEE Trans. Inf. Theory, vol. 40, no. 6, pp. 1713–1727, Nov 1994.
 \bibitem{Rappaport}
 T. Rappaport, Wireless Communications: Principles and Practice, 2nd ed. Upper Saddle River, NJ, USA: Prentice Hall PTR, 2001.
 \bibitem{Baccelli1997}
 F. Baccelli, M. Klein, M. Lebourges, and S. Zuyev, “Stochastic geometry and architecture of communication networks,” Telecommunication Systems, vol. 7, no. 1, pp. 209–227, 1997. [Online]. Available: http://dx.doi.org/10.1023/A:1019172312328
 \bibitem{Brown}
  T. X. Brown, “Cellular performance bounds via shotgun cellular systems,” IEEE J. Sel. Areas Commun., vol. 18, no. 11, pp. 2443–2455, Nov 2000.
  \bibitem{Andrews_2011}
  J. Andrews, F. Baccelli, and R. Ganti, “A tractable approach to coverage and rate in cellular networks,” IEEE Trans. Commun., vol. 59, no. 11, pp. 3122–3134, November 2011.
  \bibitem{Dhillon_2011}
  H. Dhillon, R. Ganti, and J. Andrews, “A tractable framework for coverage and outage in heterogeneous cellular networks,” Information Theory and Applications Workshop, pp. 1–6, February 2011.
  \bibitem{Dhillon2012} 
  H. S. Dhillon, R. K. Ganti, F. Baccelli, and J. G. Andrews, “Modeling and analysis of k-tier downlink heterogeneous cellular networks,” IEEE J. Sel. Areas Commun., vol. 30, no. 3, pp. 550–560, April 2012.
  \bibitem{HSJo_2012}
  H.-S. Jo, Y. J. Sang, P. Xia, and J. Andrews, “Heterogeneous cellular networks with flexible cell association: A comprehensive downlink SINR analysis,” IEEE Trans. Wireless Commun., vol. 11, no. 10, pp. 3484–3495, October 2012.
  \bibitem{Mukherjee_2011}
  S. Mukherjee, “Analysis of UE outage probability and macrocellular traffic offloading for WCDMA macro network with femto overlay under closed and open access,” IEEE Inter. Conf. on Commun., pp. 1–6, June 2011.
  \bibitem{Heath_2013}
  R. Heath, M. Kountouris, and T. Bai, “Modeling heterogeneous network interference using poisson point processes,” IEEE Trans. Signal Process., vol. 61, no. 16, pp. 4114–4126, August 2013.
  \bibitem{Haenggi_Book}
  M. Haenggi, Stochastic Geometry for Wireless Networks. Cambridge University Press, 2012, cambridge Books Online. [Online]. Available: http://dx.doi.org/10.1017/CBO9781139043816
  \bibitem{haenggi2009interference}
  M. Haenggi and R. K. Ganti, Interference in large wireless networks. Now Publishers Inc, 2009. %[Online]. Available: http://www.nd.edu/ ∼ mhaenggi/pubs/now.pdf
  \bibitem{Baccelli}
  F. Baccelli and B. Blaszczyszyn, Stochastic Geometry and Wireless Networks, Volume I - Theory, ser. Foundations and Trends in Networking.
  \bibitem{Jeffrey_2016}
  J. G. Andrews, A. K. Gupta, and H. Dhillon “A primer on cellular network analysis using stochastic geometry,” CoRR, 2016. [Online].
 Available: http://arxiv.org/abs/1604.03183
 \bibitem{Wang_2015}
  H. Wang, J. Wang, and Z. Ding, “Distributed power control in a two-tier heterogeneous network,” IEEE Trans. Wireless Commun., vol. 14, no. 12, pp. 6509–6523, Dec 2015.
   
  \bibitem{Chandrasekhar_powercontrol}
  V. Chandrasekhar, J. G. Andrews, T. Muharemovic, Z. Shen, and A. Gatherer, “Power control in two-tier femtocell networks,” IEEE Trans. Wireless Commun., vol. 8, no. 8, pp. 4316–4328, August 2009.
  \bibitem{Boudreau_2009}
  G. Boudreau, J. Panicker, N. Guo, R. Chang, N. Wang, and S. Vrzic, “Interference coordination and cancellation for 4G networks,” IEEE Commun. Mag., vol. 47, no. 4, pp. 74–81, April 2009.
  \bibitem{Huawei_2005}
  Huawei, “R1-050507: Soft frequency reuse scheme for UTRAN LTE,” 3GPP TSG RAN WG1 Meeting no. 41, Tech. Rep., May 2005.
  \bibitem{Junyi_1999}
  J. Li, N. B. Shroff, and K. P. Chong, “A reduced-power channel reuse scheme for wireless packet cellular networks,” IEEE/ACM Trans. Netw., vol. 7, no. 6, pp. 818–832, December 1999.
  \bibitem{Mahmud_2014}
  A. Mahmud and K. Hamdi, “A unified framework for the analysis of fractional frequency reuse techniques,” IEEE Trans. Commun., vol. 62, no. 10, pp. 3692–3705, October 2014.
  \bibitem{YoungjuKim_2010}
  Y. Kim, T. Kwon, and D. Hong, “Area spectral efficiency of shared spectrum hierarchical cell structure networks,” IEEE Trans. Veh. Technol., vol. 59, no. 8, pp. 4145–4151, October 2010.
  \bibitem{SFR}
  3GPP, “Soft Frequency Reuse Scheme for UTRAN LTE,” Meeting 41, Athens, Greece, Tech. Rep., May 2005.
  \bibitem{SFR1}
  Z. Xie and B. Walke, “Enhanced fractional frequency reuse to increase capacity of ofdma systems,” in 2009 3rd International Conference on New Technologies, Mobility and Security, Dec 2009, pp. 1–5.
  \bibitem{Hossain2013}
  E. Hossain, L. B. Le, and D. Niyato, “Resource allocation in two-tier networks using fractional frequency reuse,” John Wiley \& Sons, Inc., pp. 102–122, 2013. %[Online]. Available: http://dx.doi.org/10.1002/9781118749821.ch5
  \bibitem{Giuliano_2008}
  R. Giuliano, C. Monti, and P. Loreti, “WiMAX fractional frequency reuse for rural environments,” IEEE Wireless Communications, vol. 15, no. 3, pp. 60–65, June 2008.
  \bibitem{Chang_2013}
  R. Y. Chang, Z. Tao, J. Zhang, and C. J. Kuo, “Dynamic fractional frequency reuse (D-FFR) for multicell OFDMA networks using a graph framework,” Wirel. Commun. Mob. Comput., vol. 13, no. 1, pp. 12–27, January 2013.
  \bibitem{ElSawy}
  H. ElSawy, E. Hossain, and M. Haenggi, “Stochastic geometry for modeling, analysis, and design of multi-tier and cognitive cellular wireless networks: A survey,” IEEE Communications Surveys Tutorials, vol. 15, no. 3, pp. 996–1019, 2013.
  \bibitem{Novlan_2011} 
  T. Novlan, R. Ganti, A. Ghosh, and J. Andrews, “Analytical evaluation of fractional frequency reuse for OFDMA cellular networks,” IEEE Trans. Wireless Commun., vol. 10, no. 12, pp. 4294–4305, December 2011.
  \bibitem{Novlan_2012}
  ——, “Analytical evaluation of fractional frequency reuse for heterogeneous cellular networks,” IEEE Trans. Commun., vol. 60, no. 7, pp. 2029–2039, July 2012.
  \bibitem{Zhuang2014}
  H. Zhuang and T. Ohtsuki, “A model based on poisson point process for downlink K tiers fractional frequency reuse heterogeneous networks,” Physical Communication, vol. 13, Part B, pp. 3 – 12, 2014, Special Issue on Heterogeneous and Small Cell Networks.
  \bibitem{Dhillon_2013}
  H. Dhillon, R. Ganti, and J. Andrews, “Load-aware modeling and analysis of heterogeneous cellular networks,” IEEE Trans. Wireless Commun., vol. 12, no. 4, pp. 1666–1677, April 2013.
  \bibitem{Wei_Bao_2014_NearOptimal}
  W. Bao and B. Liang, “Near-optimal spectrum allocation in multi-tier cellular networks with random inelastic traffic,” ICASSP, pp. 855–859, May 2014.
  \bibitem{Praful_BlockingProb}
  P. Mankar, B. Sahu, G. Das, and S. Pathak, “Evaluation of blocking probability for downlink in poisson networks,” IEEE Wireless Commun. Lett., vol. PP, no. 99, pp. 1–1, 2015.
  \bibitem{Kleinrock}
  L. Kleinrock, Queueing Systems, Volume 1: Theory. Wiley-Interscience, 1975.
  \bibitem{Moltchanov}
  D. Moltchanov, “Distance distributions in random networks,” Ad Hoc Networks, vol. 10, no. 6, pp. 1146 – 1166, 2012.
  \bibitem{Bertsekas:1992}
  D. Bertsekas and R. Gallager, Data Networks (2Nd Ed.). Upper Saddle River, NJ, USA: Prentice-Hall, Inc., 1992.
  \bibitem{Singh_2013}
  S. Singh, H. Dhillon, and J. Andrews, “Offloading in heterogeneous networks: Modeling, analysis, and design insights,” IEEE Trans. Wireless Commun., vol. 12, no. 5, pp. 2484–2497, May 2013.
  \bibitem{Kaufman}
  J. S. Kaufman, “Blocking in a shared resource environment,” IEEE Trans. Commun., vol. 29, no. 10, pp. 1474–1481, Oct 1981.
  \bibitem{Roberts}
  J. W. Roberts, “A service system with heterogeneous user requirements,” Performance of Data Communications Systems and their Applications, 1981.
  \bibitem{Karray_2010}
  M. K. Karray, “Analytical evaluation of qos in the downlink of OFDMA wireless cellular networks serving streaming and elastic traffic,” IEEE Trans. Wireless Commun., vol. 9, no. 5, pp. 1799–1807, May 2010.
  \bibitem{KeithRoss}
  K. W. Ross, Multiservice Loss Models for Broadband Telecommunication Networks, Springer

\end{thebibliography}

\begin{IEEEbiography}[{\includegraphics[width=1in,height=1.25in,clip,keepaspectratio]{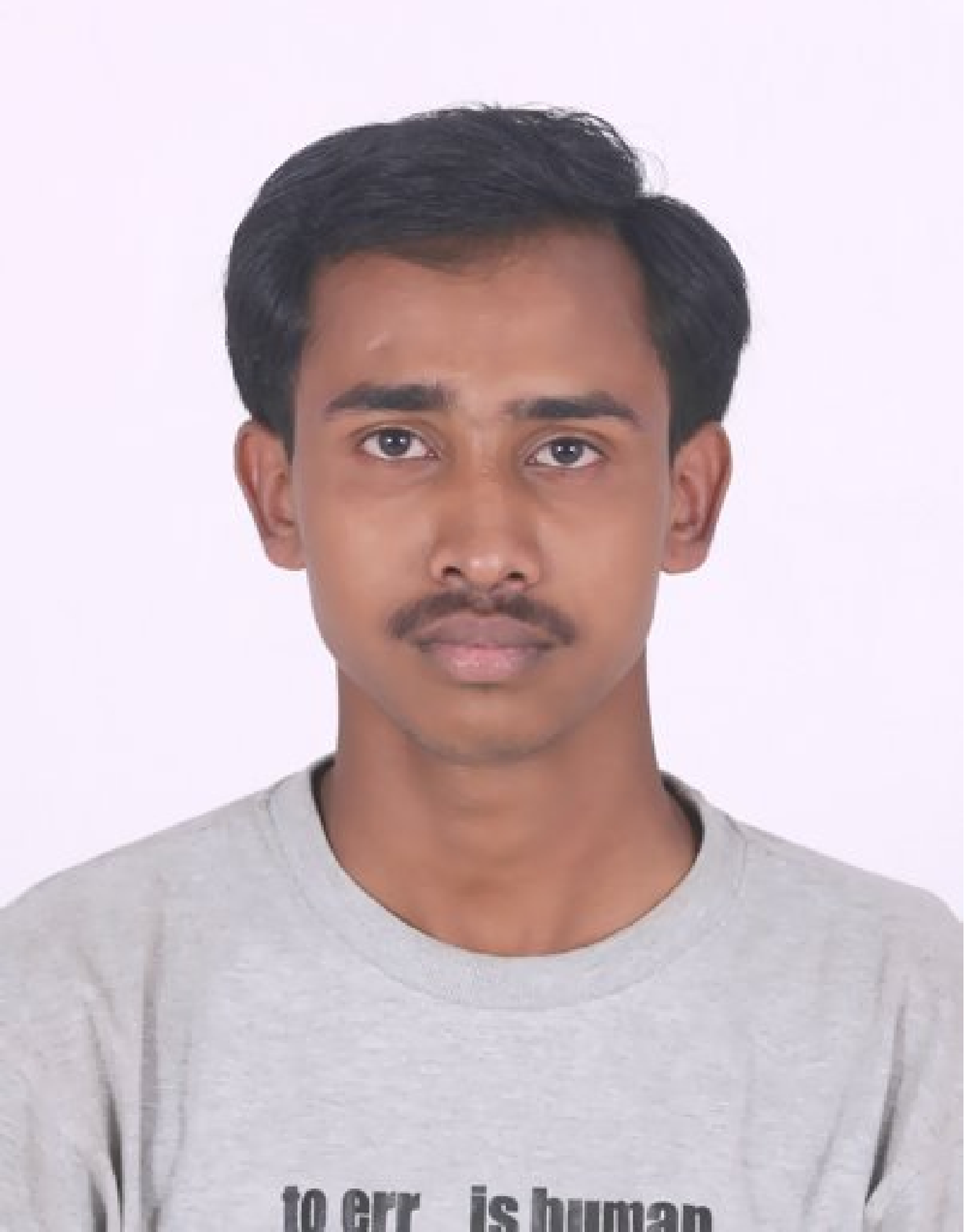}}]{Praful D. Mankar} has
completed B.Tech degree in Electronics and Telecommunications Engineering from Amaravati University, MH, India in 2006. In 2009, he has obtained M.Tech degree in Telecommunication System Engineering  from Indian Institute of Technology (IIT) Kharagpur, WB, India. He received Ph.D degree in wireless communication from IIT Kharagpur, WB, India in 2016.  Currently, He is working as a research assistant at IIT Kharagpur, WB, India. His research interest includes modeling, analyzing and designing of wireless networks. 
\end{IEEEbiography}

\begin{IEEEbiography}[{\includegraphics[width=1in,height=1.25in,clip,keepaspectratio]{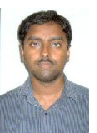}}]{Goutam Das} has obtained his Ph.D. Degree from the University of Melbourne, Australia in 2008. He has worked as a postdoctoral fellow at Ghent University, Belgium, from 2009-2011. Currently, he is working as an Assistant Professor in the Indian Institute of Technology, Kharagpur. He has served as a member in the organizing committee of IEEE ANTS since 2011. His research interests include optical access networks, optical data center networks, radio over fiber technology, optical packet switched networks and media access protocol design for application specific requirements.
\end{IEEEbiography}

\begin{IEEEbiography}[{\includegraphics[width=1in,height=1.25in,clip,keepaspectratio]{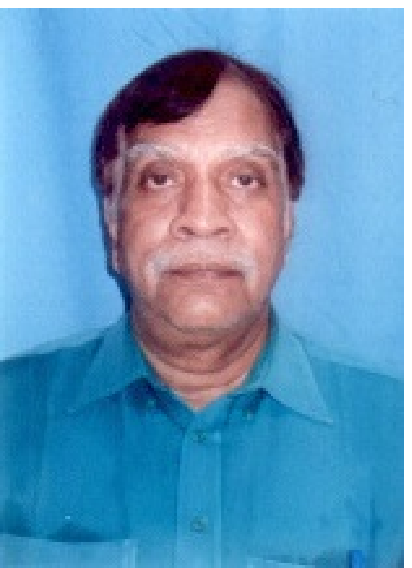}}]{Sant Sharan Pathak} has received his B.Tech and M.Tech Degrees in Electronics Engineering from IT BHU in 1976 and 1978 respectively,
and Ph.D. Degree in Digital Communications from IIT Delhi in 1984. He has joined the Department of Electronics and Electrical
Communication Engineering in 1985. His area of research interest includes physical layer network issues, receiver design optimization for Gaussian and non-Gaussian channels, security system design and analysis at application layer, image forensics with wireless camera pickup over Internet using low bandwidth wireless access channel, and similar others.
\end{IEEEbiography}
\end{document}